\documentclass[final,3p,times]{elsarticle}
\usepackage[centertags]{amsmath}
\usepackage{amsfonts}
\usepackage{amssymb}
\usepackage{amsthm}
\usepackage{mathrsfs} %script fonts
\usepackage{epsfig,multicol,bbm}
\usepackage{graphicx}
\usepackage[nodots,nocompress]{numcompress}

\journal{ArXiv}

\def\supp{\hbox{supp}}

\begin{document}
\hyphenation{Co-lom-beau}

\pagestyle{plain}

\begin{frontmatter}

\title{New ideas about multiplication of tensorial distributions}

\author{Jozef Skakala}

\ead{Jozef.Skakala@msor.vuw.ac.nz}

\address{School of Mathematics, Statistics and Operations Research\\Victoria University of Wellington, New Zealand}
\date{1 August 2009; \LaTeX-ed \today}

\begin{abstract}

There is a need in general relativity for a consistent and
useful mathematical theory defining the multiplication of tensor
distributions in a geometric (diffeomorphism invariant) way.
Significant progress has been made through the concept of Colombeau
algebras, and the construction of full Colombeau algebras on
differential manifolds for arbitrary tensors. Despite the fact that
this goal was achieved, it does not incorporate clearly enough the concept of
covariant derivative and hence is of a limited use.
We take a different approach: we consider any type of preference for
smooth distributions (on a smooth manifold) as nonintuitive, which
means all our approach must be based fully on the Colombeau
equivalence relation as the fundamental feature of the theory. After
taking this approach we very naturally obtain canonical and
geometric theory defining tensorial operations with tensorial
distributions, including covariant derivative. This also happens
because we no longer need the
construction of Colombeau algebras. The important advantage of our
approach lies also in the fact that it brings physical insight
into the mathematical concepts used and naturally leads to
formulation of physics on (what we call) piecewise smooth manifolds,
rather than on smooth manifold. This brings to the language
additional symmetry (in the same way as turning from Poincare
invariance to diffeomorphism invariance), and is compatible with our
intuition that ``pointwise'' properties in some specific sense
``do not matter''.

\end{abstract}


\end{frontmatter}

\tableofcontents

\bigskip

\section{Introduction}
%%%%%%%%%%%%%%%%%%%%%%%%%%%%%%%%%%%%%%%%%%%%%%%%%%%

This paper is devoted to a topic from the field of mathematical physics which is closely related to the general theory of relativity. It offers possible significant conceptual extensions of general relativity at short distances/high energies, and gives another arena in which one can conceptually/physically modify the classical theory.
But the possible meaning of these ideas is much wider than just the general theory of relativity. It is related to the general questions of how one defines the theory of distributions for any \emph{geometrically} formulated physics describing \emph{interactions}.

\paragraph{The main reasons why we ``bother''}
Let us start by giving the basic reasons why one should work with the language of distributions rather than with the old language of functions:

\begin{itemize}
\item
First there are deep physical reasons for working with distributions
rather than with smooth tensor fields. We think distributions are
more than just a convenient tool for doing computations in those
cases in which one cannot use standard differential geometry. We
consider them to be mathematical objects which much more accurately
express what one actually measures in physics experiments, more so
than when we compare them to the old language of smooth functions.
The reason is that the question: ``What is the `amount' of physical
quantity contained in an open set?'' is in our view a much more
reasonable physical question, (reasonable from the point of view of
what we can ask the experimentalists to measure), than the question:
``What is the value of quantity at a given point?''. But ``point
values'' as ``recovered'' by delta distributions do seem to give a
precise and reasonable meaning to the last question. There is also a
strong intuition that the ``amount'' of a physical quantity in the
open set $\Omega_{1}\cup \Omega_{2}$, where $\Omega_{1},\Omega_{2}$
are disjoint is the sum of the ``amounts'' of that quantity in
$\Omega_{1}$ and $\Omega_{2}$. That means it is more appropriate to
speak about distributions rather than general smooth mappings from
functions to the real numbers (the mappings should also be linear).

\item
The second reason is that many physical applications suggest the need for a much
richer language than the language of smooth tensor fields. Actually when we look for physically interesting solutions it might be always a matter of importance to have a much larger class of objects available than the class of smooth tensor fields.

\item
The third reason (which is a bit more speculative) is the relation of
the language defining the multiplication of distributions to quantum field theories (but specifically to quantum gravity). Note that the problems requireing distributions, that means problems going
beyond the language of classical differential geometry, might be related to physics on small scales. At the same time understanding some operations with distributions, specifically their product has a large \emph{formal} impact on the foundations of quantum field theory, particularly on the
problem with interacting fields. (See, for example, \cite{QFT}.) As a result of this it can have significant consequences for quantum gravity as well.
\end{itemize}

\paragraph{The intuition behind our ideas}
Considerations about language intuitiveness lead us to an interesting conclusion: the language of distributions (being connected with our intuition)~
strongly suggests that the properties of classical
tensor fields should not depend on the sets having (in every chart) Lebesgue measure 0. If we follow the idea that these measure zero sets  do not have any impact on physics, we should naturally expect that we will be able to
generalize our language from a smooth manifold into a piecewise
smooth manifold (which will bring higher symmetry to our conceptual network).
The first traces of piecewise smooth coordinate transformations are
already seen, for example, in~\cite{Penrose}.

\paragraph{The current situation is ``strange''}
Now it is worth noting how strange the current situation of the theory of distributions is: we have a useful and
meaningful language of distributions, which can be geometrically
generalized, but this language works only for linear physics. But linear physics
is only a starting point (or at best a rough approximation) for
describing real physical interactions, and hence nonlinear physics. So one
naturally expects that the ``physical'' language of distributions will be a result of some
mathematical language defining their product. Moreover,
at the same time we want this language to contain the old language of differential geometry (as a special case), as in the case of linear
theories. It is quite obvious that the nonlinear generalization of the geometric distributional theory and the construction of generalized differential geometry are just two routes to the same mathematical theory. The natural feeling that such theory might exist is the main motivation for this work. The practical need of this language is obvious as well, as we
see in the numerous applications \cite{Kerr, GerTra, QFT, Schwarzschild, geodesic, Reis-Nor, waves, Vickers}. (Although this is not the main
motivation of our work.) But is it necessary that such more general mathematical language exists as a full well defined theory? No, not at all. The potentially successful uses of distributions that go behind the Schwarz original theory might be only ``ad hoc'' from the fundamental reasons. Take the classical physics. The success of such uses of distributions here might not be a result of some more general language than smooth tensor calculus being a classical limit of the more fundamental physics. There might be only hidden specific reasons in the more fundamental physics why such ``ad hoc'' calculations in some of the particular cases work. But it is certainly very interesting and important to explore and answer the question whether: (i) such a full mathematical langauge exists, and (ii) to which extent it is useful to the physics community. For the first question this work suggests a positive answer, to answer the second question much more work has to be done.

\paragraph{So what did we achieve?}
The main motivation for this work is the development of
a language of distributional tensors, strongly connected with physical intuition. (This also means it has to be based on the concept of a piecewise smooth manifold.) This language must contain all the basic tensorial operations in a generalized way, enabling us to understand
the results the community has already achieved, and also the problems attached to them. It is worth to stress that some of this motivation results from a shift in view regarding the foundations of classical physics (so it is given by ``deeper'' philosophical reasons), but it can have also an impact on practical physical questions. The scale of this impact has still to be explored. We claim that the goals defined by our motivation (as described at the beginning) are achieved in this work. Particularly, we have generalized all the basic concepts from smooth tensor field calculus (including the fundamental concept of the covariant derivative) in two basic directions:

\begin{itemize}
\item
The first generalization goes in the direction of the class of objects that, (in every chart on the piecewise smooth manifold), are indirectly related to sets of piecewise continuous functions. This class we call $D'^{m}_{n E A}(M)$. For a detailed understanding see section \ref{GTFsection}.
\item
The second generalization is a generalization to the class of objects naturally connected with a smooth manifold belonging to our piecewise smooth manifold (in the sense that the smooth atlas of the smooth manifold is a subatlas of our piecewise smooth atlas). This is a good analogy to the generalization known from the classical distribution theory. The class of such objects we call $D'^{m}_{n (\mathcal{\mathcal{\mathcal{\mathcal{S}}}} o)}(M)$. For detailed understanding see again the section \ref{GTFsection}.
\end{itemize}
We view our calculus as the most natural and straightforward
construction achieving these two particular goals. The fact that
such a natural construction seems to exist supports our faith in the
practical meaning of the mathematical language here developed.

The last goal of this paper is to suggest much more ambitious, natural generalizations, which are unfortunately at present only in the form of conjectures. Later in the text we provide the reader with such conjectures.

\paragraph{The structure of this paper}
The structure of this paper is the following:
In the first part we want to present the
current state of the Colombeau algebra theory and its geometric
formulations. We want to indicate where its
weaknesses are.
This part is followed by several technical sections, in which we define our theory and prove the basic theorems. First we define the basic concepts underlying our theory. After that we define the concept of generalized tensor fields, their important subclasses and basic operations on the generalized tensor fields (like tensor product). This is followed by the definition of the basic concept of our theory, the concept of equivalence between two generalized tensor fields. The last technical part deals with the definition of the covariant derivative operator and formulation of the initial value problem in our theory.
All these technical parts are followed by explanatory sections, where we discuss our results and show how our theory relates to the practical results already achieved (as described in the first part of the paper).

\section{Overview of the present state of the theory}

\subsection{The theory of Schwartz distributions}

Around the middle of the 20-th century Laurent Schwartz found a mathematically rigorous way
for extending the language of physics from the
language of smooth functions into the language of distributions. Physicists such as Heaviside and Dirac had already given good physics reasons for believing that such a mathematical structure might exist.

The classical formulations of the distribution theory were directly connected with
$\mathbb{R}^{n}$ and were non-geometric. Distributions in such a formulation are typically understood as continuous, linear maps
from compactly supported smooth functions (on $\mathbb{R}^{n}$) to
real numbers. (The class of such functions is typically denoted by $D(\mathbb{R}^{n})$. The dual to such space, which is what distributions are, is typically denoted by $D'(\mathbb{R}^{n})$.) Here the word ``continuous'' refers to the following topology on
the given space of compactly supported smooth functions:
\emph{The sequence of smooth compactly supported functions $f_{l}(x_{i})$ converges
for $l\to\infty$ to a smooth compactly supported function $f(x_{i})$ iff an arbitrary
degree derivative with respect to arbitrary variables $\frac{\partial^{n}f_{l}(x_{i})}{\partial x_{1}^{m_{1}}...\partial x_{k}^{m_{k}}}$ $\left(\sum_{i=1}^{k}m_{k}=n\right)$
converges uniformly on each compact $\mathbb{R}^{n}$ subset to $\frac{\partial^{n}f(x_{i})}{\partial x_{1}^{m_{1}}...\partial x_{k}^{m_{k}}}$.}

The alternative classical way of formulating the theory of distributions is to extend the space of test objects to be such that they still follow appropriate fall-off properties. The properties can be summarized as:

$f(x_{i})$ belongs to such space if it is smooth and
\begin{eqnarray}
(\forall\alpha,~\forall\beta)~~~~~~ \lim_{x_{1},\dots x_{k}\to\infty}\left(x_{1}^{n_{1}}\dots x_{k}^{n_{k}}\frac{\partial^{\alpha}f(x_{i})}{\partial x_{1}^{m_{1}}...\partial x_{k}^{m_{k}}}\right)=0,\\
\sum_{i=1}^{k}n_{i}=\beta, ~~~ \sum_{j=1}^{k}m_{j}=\alpha.~~~~~~~~~~~~~~~~~~~~~~~~~~~~\nonumber
\end{eqnarray}
This is a topological space with topology given by a set of semi-norms $P_{\alpha,\beta}$ defined\footnote{We admit that the notation $P_{\alpha,\beta}$ might be somewhat misleading, since the $\alpha,\beta$ values do not specify the semi-norm in a unique way.} as
\begin{eqnarray}
P_{\alpha,\beta}=\sup_{x_{i}}\left|x_{1}^{n_{1}}\dots x_{k}^{n_{k}}\frac{\partial^{\alpha}f(x_{i})}{\partial x_{1}^{m_{1}}...\partial x_{k}^{m_{k}}}\right|,\\
\sum_{i=1}^{k}n_{i}=\beta, ~~~ \sum_{j=1}^{k}m_{j}=\alpha.~~~~~~~~~~~~~~~~~~~\nonumber
\end{eqnarray}
Naturally, the space of continuous and linear maps on such a space is more restricted as in the first, more common formulation. Objects belonging to such duals are called ``tempered distributions''. The advantage of this more restricted version is that Fourier transform is a well defined mapping on the space of tempered distributions.

If we refer to the more common, first formulation, the space of distributions accommodates the linear
space of smooth functions by the mapping:

\begin{equation}\label{map}
f(x_{i})\to\int_{\mathbb{R}^{n}}f(x_{i})\cdots d^{n}x.
\end{equation}

Here $f(x_{i})$ is a smooth function on $\mathbb{R}^{n}$ and
$\int_{\mathbb{R}^{n}}f(x)\cdots dx$ is a distribution defined as the
mapping:

\begin{equation}
\Psi(x_{i})\to\int_{\mathbb{R}^{n}}f(x_{i})\Psi(x_{i})~d^{n}x, ~~~~\Psi(x_{i})\in D(\mathbb{R}^{n}).
\end{equation}

Moreover, by use of this mapping one can map into the space of
distributions any Lebesgue integrable function (injectively up to a
function the absolute value of which has Lebesgue integral 0) . The
distributional objects defined by the images of the map \eqref{map}
are called regular distributions. The map \eqref{map} always
preserves the linear structure, hence the space of Lebesgue integrable functions is a linear
subspace of the space of distributions. But the space of distributions
is a larger space than the space of regular distributions. This can be
easily demonstrated by defining the delta distribution

\begin{equation}
\delta(\Psi)\equiv\Psi(0)
\end{equation}
and showing that such mapping
cannot be obtained by a regular distribution. Note particularly that delta
distribution had an immediate use in physics in the description of
point-like sources. (Its intuitive use in the work of Paul A.M. Dirac before the theory of Schwartz distributions was found, was one of the main physics reasons why people searched for such language extension.)

These considerations show that the space of distributions is \emph{significantly} lar\-ger than the space
of smooth functions. The space of distributions is classically taken to be a
topological space with the weak ($\sigma-$) topology. In this
topology the space of regular distributions given by smooth
functions is a dense set.

Moreover one can continuously extend the derivative operator from the space
of $C^{1}(\mathbb{R}^{n})$ functions to the space of distributions by using the definition:

\begin{equation}
T_{,~x^{i}}(\Psi)\equiv -T(\Psi_{,~x^{i}}),
\end{equation}
(where $T$ denotes a distribution). This means that $C^{\infty}(\mathbb{R}^{n})$ functions form not only a
linear subspace in the space of distributions, but also a differential
linear subspace. It looks like there stands ``almost'' nothing in the
way of fully and satisfactorily extending the language of
$C^{1}(\mathbb{R}^{n})$ functions to the language of Schwartz
distributions. Unfortunately there is still one remaining trivial
operation and this is the operation of multiplication. That means we
need to obtain some distributional algebra having as subalgebra the algebra of
Lebesgue integrable functions factorized by functions the absolute
value of which has Lebesgue interal 0.
Unfortunately, shortly after Laurent Schwartz formulated the theory
of distributions he proved the following ``no-go'' result \cite{Schwarz}:

The requirement of constructing an algebra that fulfills the following three conditions is inconsistent:

\begin{itemize}
\item[a)] the space of distributions is linearly embedded into the
algebra,
\item[b)] there exists a linear derivative operator, which fulfills
the Leibniz rule and reduces on the space of distributions to the
distributional derivative,
\item[c)] there exists a natural number $k$ such, that our algebra has as subalgebra the algebra of $C^{k}(\mathbb{R}^{n})$
functions.
\end{itemize}
This is called the Schwartz impossibility theorem. There is a nice
example showing where the problem is hidden: Take the Heaviside
distribution $H$. Suppose that $a)$ and $b)$ hold and we multiply
the functions/distributions in the usual way. Then since ~$H^{m}=H$ the following must hold:

\begin{equation}
H'=(H^{m})'=\delta=mH^{m-1}\delta~.
\end{equation}
But this actually implies that $\delta=0$, which is nonsense.

The closest one can get to fulfill the conditions $a)-c)$ from the
Schwarz impossibility theorem is the Colombeau algebra (as defined in the next section) where
conditions $a)- c)$ are fulfilled with the exception that in the $c)$
condition $k$ is taken to be infinite. This means that only the smooth
functions form a subalgebra of the Colombeau algebra. This is obviously not
satisfactory, since we know (and need to recover) rules for multiplying
multiply much larger classes of functions than only smooth functions. In the case of Colombeau algebras this problem is ``resolved'' by
the equivalence relation, as we will see in the following
section.

\subsection{The standard $\mathbb{R}^{n}$ theory of Colombeau
algebras}

\subsubsection{The \emph{special} Colombeau algebra and the embedding of distributions}
The so called special Colombeau algebra is on
$\mathbb{R}^{n}$ defined as:
\begin{equation}
\mathcal{G}(\mathbb{R}^{n})=
\mathcal{E}_{M}(\mathbb{R}^{n})/\mathcal{N}(\mathbb{R}^{n}).
\end{equation}
Here $\mathcal{E}_{M}(\mathbb{R}^{n})$ (moderate functions) is
defined as the algebra of functions:
\begin{equation}
\mathbb{R}^{n} \times(0,1]\to\mathbb{R}
\end{equation}
 that are smooth on $\mathbb{R}^{n}$ (this is usually called $\mathcal{E} (\mathbb{R}^{n})$), and for any compact subset $K$ of ~$\mathbb{R}^{n}$ (for which we will henceforth use the notation ~$K\subset\subset\mathbb{R}^{n}$) it  holds that:
\begin{equation}
\forall\alpha\in\mathbb{N}^{n}_{0},~\exists p\in\mathbb{N} ~\hbox{
such that}\footnote{In this definition, (and also later in the text), the symbol $\mathbb{N}^{n}_{0}$ denotes a sequence of $n$ members formed of natural numbers (with 0 included).}  ~\sup_{x\in K}|D^{\alpha}f_{\epsilon}(x_{i})|\leq
O(\epsilon^{-p}) ~\hbox{ as } ~\epsilon\to
 0.
\end{equation}
The $\mathcal{N}(\mathbb{R}^{n})$ (negligible functions) are
functions from $\mathcal{E}(\mathbb{R}^{n})$ where for any
$K\subset\subset\mathbb{R}^{n}$ it holds that:
\begin{equation}
\forall\alpha\in\mathbb{N}^{n}_{0},~\forall p\in\mathbb{N} ~\hbox{
we have } ~\sup_{x_{i}\in K}|D^{\alpha}f_{\epsilon}(x_{i})|\leq
O(\epsilon^{p}) ~\hbox{ as } ~\epsilon\to 0.
\end{equation}
 The first definition tells us that moderate functions are those whose partial derivatives of arbitrary degree (with respect to variables $x_{i}$) do not diverge faster then any arbitrary negative power of $\epsilon$, as $\epsilon\to 0$. Negligible functions are those moderate functions whose partial derivatives of arbitrary degree go to zero faster than any positive power of $\epsilon$, as $\epsilon\to 0$. This simple formulation can be straightforwardly generalized
into general manifolds just by substituting the concept of Lie derivative for the ``naive'' derivative
used before.

It
can be shown, by using convolution with an arbitrary smoothing
kernel (or mollifier), that we can embed a distribution into the Colombeau algebra. By a
smoothing kernel we mean, in the widest sense a compactly supported, smooth
function $\rho_{\epsilon}(x_{i})$, with $\epsilon\in (0,1]$, such that:
\begin{itemize}
\item
$\supp(\rho_{\epsilon})\to\{0\}$ ~~for ~~($\epsilon\to 0$),
\item
$\int_{\mathbb{R}^{n}}\rho_{\epsilon}(x_{i})~d^{n}x\to 1$~~~ for ~~($\epsilon\to 0$),
\item
$\forall\eta >0$ ~~$\exists C$,~~
$\forall\epsilon\in (0,\eta)$~~~ $\sup_{x_{i}}
|\rho_{\epsilon}(x_{i})|<C$.
\end{itemize}
This most generic embedding approach is mentioned for example in
\cite{Penrose}~(in some sense also in \cite{geodesic}).

More ``restricted'' embeddings to $\mathcal{G}(\mathbb{R}^{n})$ are
also commonly used. We can choose for instance a subclass of
mollifiers called ~$A^{0}(\mathbb{R}^{n})$, which are smooth
functions from ~$D(\mathbb{R}^{n})$ (smooth, compactly supported)
and (i.e. \cite{Gsponer}) such that

\begin{equation}
\forall\epsilon~~\hbox{holds}~~\int_{\mathbb{R}^{n}}\rho_{\epsilon}(x_{i})d^{n}x =
 1.
\end{equation}
Their dependence on $\epsilon$ is given\footnote{Later in the text will the notation $\rho_{\epsilon}(x_{i}), \psi_{\epsilon}(x_{i})$ (etc.) automatically mean dependence on the variable $\epsilon$ as in \eqref{epsilon}.} as
\begin{equation}\label{epsilon}
\rho_{\epsilon}(x_{i})\equiv \frac{1}{\epsilon^{n}}\rho\left(\frac{x_{i}}{\epsilon}\right)~.
\end{equation}

Sometimes the class is even more restricted. To
obtain such a formulation, we shall define classes ~$A^{m}(\mathbb{R}^{n})$ as classes
of smooth, compactly supported functions, such that
\begin{equation}
 \int_{\mathbb{R}^{n}}
x^{i}_{1}\cdots x^{j}_{l}~\phi(x_{i})d^{n}x=\delta_{0k}~~ \hbox{ for }~~
i+\cdots+j=k\leq m.
\end{equation}
Clearly $A^{m+1}(\mathbb{R}^{n})\subset
A^{m}(\mathbb{R}^{n})$. Then the most restricted class of mollifiers
is taken to be the class $A^{\infty}(\mathbb{R}^{n})$. This approach
is taken in the references
\cite{Schwarzschild, Oberguggenberger, Reis-Nor, waves, Vickers}.

Even in the case of the more restricted class of mollifiers the embeddings are generally non-canonical \cite{Gsponer, Vickers}. The exception are smooth distributions, where the difference between two embeddings related to two different mollifiers is always a negligible function.

\subsubsection{The \emph{full} Colombeau algebra and the embedding of distributions}

What is usually considered
to be the canonical formulation of Colombeau algebras in
$\mathbb{R}^{n}$ is the following: The theory is formulated in terms
of functions
\begin{equation}
\mathbb{R}^{n} \times A^{0}(\mathbb{R}^{n})\to\mathbb{R} ~~~\hbox{ (call
them $F$) }.
\end{equation}
The Colombeau algebra is defined in such way that it is a factor
algebra of moderate functions over negligible functions, where:

\begin{itemize}
\item
Moderate functions are functions from $F$ that satisfy:
\begin{eqnarray}
\forall m\in \mathbb{N}^{n}_{0}, ~\forall K\subset\subset\mathbb{R}^{n} ~~\exists N\in\mathbb{N} ~\hbox{ such that
if } ~\phi\in A^{N}(\mathbb{R}^{n}), \hbox{ there are } \alpha, \rho
>0, \nonumber\\  \hbox{ such that }   \sup_{x_{i}\in K}~\left
|D^{m}F(\phi_{\epsilon},x_{i})\right|\leq
\alpha\epsilon^{-N} ~\hbox{ if }~ 0<\epsilon<\rho.~~~~~~~~~~~~~~~~
\end{eqnarray}
\item
Negligible functions are functions from $F$ that satisfy:
\begin{eqnarray}
  \forall m\in\mathbb{N}, ~\forall K\subset\subset\mathbb{R}^{n},~~\forall p\in\mathbb{N} ~~\exists q\in\mathbb{N} ~\hbox { such that if } ~~\phi\in A^{q}(\mathbb{R}^{n}), ~\exists\alpha,\rho>0, \nonumber\\ \hbox { we have } \sup_{x_{i}\in K}\left|D^{m}F(\phi_{\epsilon},x_{i})\right|\leq\alpha\epsilon^{p} ~\hbox { if }~ 0<\epsilon<\rho.~~~~~~~~~~~~~~~~~~~~~~~~
\end{eqnarray}

\end{itemize}

Then ordinary distributions automatically define such functions by
the convolution (\cite{multi, Damyanov, geodesic, Vickers}
etc.):\footnote{Here the ``$B_{x}$'' notation means that $x$ is the
variable removed by applying the distribution.}
\begin{equation}
 B\to
B_{x}\left[\frac{1}{\epsilon^{n}}\;
\phi\left(\frac{y_{i}-x_{i}}{\epsilon}\right)\right], ~~\phi\in
A^{0}(\mathbb{R}^{n}).
\end{equation}

\subsubsection{Important common feature of both formulations}

All of these formulations have two important consequences. Given that $C$ denotes the embedding mapping:
\begin{itemize}
\item
Smooth functions ($C^{\infty}(\mathbb{R}^{n})$) form a
subalgebra of the Colombeau algebra
($C(f)C(g)=C(f \cdot g)$~ for $f$, $g$ being
smooth distributions).
\item
Distributions form a differential linear subspace of Colombeau algebra
(this means
for instance that $C(f')=C'(f)$~).
\end{itemize}

\subsubsection{The relation of equivalence in the \emph{special} Colombeau algebra}

We can formulate a relation of equivalence between an element of the
special Colombeau algebra $f_{\epsilon}(x_{i})$ and a distribution $T$. We call them equivalent,
if for any ~$\phi\in D(\mathbb{R}^{n})$,~ we have
\begin{equation}
\lim_{\epsilon\to
0}\int_{\mathbb{R}^{n}} f_{\epsilon}(x_{i})\phi(x_{i})d^{n}x= T(\phi).
\end{equation}
 Then two elements of
Colombeau algebra $f_{\epsilon}(x_{i})$, $g_{\epsilon}(x_{i})$ are equivalent,
if for any ~$\phi(x_{i})\in D(\mathbb{R}^{n})$
\begin{equation}
 \lim_{\epsilon\to
0}\int_{\mathbb{R}^{n}} \left(f_{\epsilon}(x_{i})-g_{\epsilon}(x_{i})\right)\phi(x_{i})~d^{n}x=0.
\end{equation}
For the choice of
$A^{\infty}(\mathbb{R}^{n})$ mollifiers the following relations are respected by
the equivalence:
\begin{itemize}
\item
It respects multiplication of distribution by a
smooth distribution \cite{Vickers} in the sense that: ~~$C(f\cdot g)\approx
C(f)\cdot C(g)$,~ where ~$f$ is a smooth distribution and $g\in D'(\mathbb{R}^{n})$.~
\item
It
respects (in the same sense) multiplication of piecewise continuous
functions (we mean here regular distributions given by
piecewise continuous functions) \cite{multi}.
\item
If $g$ is a distribution and $f\approx g$,
then for arbitrary natural number $n$~ it holds ~$D^{n}f\approx D^{n}g$ ~\cite{Damyanov}.
\item
If $f$ is
equivalent to distribution $g$, and if $h$ is a smooth
distribution, then $f\cdot h$ is equivalent to $g\cdot h$ ~\cite{Damyanov}.
\end{itemize}

\subsubsection{The relation of equivalence in the \emph{full} Colombeau algebra}

In the canonical formulation the
equivalence relation is again formulated either between an element of the Colombeau
algebra and a distribution, or analogously between two elements of
the Colombeau algebra:
If there  $\exists m$, such that for any
$\phi\in A^{m}(\mathbb{R}^{n})$, and for any
~$\Psi(x_{i})\in D(\mathbb{R}^{n})$, it holds that
\begin{equation}
\lim_{\epsilon\to 0}\int_{\mathbb{R}^{n}}
(f(\phi_{\epsilon},x_{i})-g(\phi_{\epsilon},x_{i}))\Psi(x_{i})~d^{n}x=0,
\end{equation}
then we say that $f$ and $g$ are equivalent ($f\approx g$).
For the canonical
embedding and differentiation we have the same commutation relations as in the
non-canonical case. It can be also
proven that for $f_{1}\dots f_{n}$ being regular distributions given
by piecewise continuous functions it follows that
\begin{equation}
C(f_{1})...C(f_{n})\approx C(f_{1}...f_{n}),
\end{equation}
 and for $f$ being
arbitrary distribution and $g$ smooth distribution it holds that
\begin{equation}
C(f)\cdot C(g)\approx C(f\cdot g).
\end{equation}

\subsubsection{How this relates to some older Colombeau papers}

In
older Colombeau papers \cite{multi, Amulti} all these concepts are formulated (equivalently)
as the relations between elements of a Colombeau algebra taken as a
subalgebra of the ~$C^{\infty}(D(\mathbb{R}^{n}))$ algebra. The definitions of
moderate and negligible elements are almost exactly the same as in
the canonical formulation, the only difference is that their domain is taken here to be the class
~$D(\mathbb{R}^{n})\times\mathbb{R}^{n}$~ (being a larger domain than $A^{o}(\mathbb{R}^{n})\times\mathbb{R}^{n}$). The canonical formulation is related to the elements of the class $C^{\infty}(D(\mathbb{R}^{n}))$ through their convolution with the objects from the class
$D(\mathbb{R}^{n})$.  It is easy to see that you can
formulate all the previous relations as relations between elements of the ~$C^{\infty}(D(\mathbb{R}^{n}))$
subalgebra (with pointwise multiplication), containing also distributions.

\subsection{Distributions in the geometric approach}

This part is devoted to review the distributional theory in the geometric framework. How to define arbitrary rank tensorial distribution
on arbitrary manifolds by avoiding reference to preferred charts? Usually we mean by a
distribution representing an $(m,n)$ tensor field an element from the
dual to the space of objects given by the tensor product of $(m,n)$ tensor fields and smooth
compactly supported $k$-form fields
 (on $k$ dimensional space). That means for example a regular distributional $(m,n)$ tensor field $B^{\mu...}_{\nu...}$ is introduced as a map
\begin{equation}
T\otimes\omega\to\int B^{\nu...\beta}_{\mu...\alpha}~T^{\mu...\alpha}_{\nu...\beta}~\omega ~.
\end{equation}
This is very much the same as to say that the test space are smooth
compactly supported tensor densities
$T^{\mu....\alpha}_{\nu...\beta}$
\cite{Garfinkle, GerTra}. The topology taken on this space is the
usual topology of uniform convergence for arbitrary derivatives
related to arbitrary charts (so the convergence from $\mathbb{R}^{n}$ theory should be valid in
all charts). The derivative operator acting on this space is typically Lie derivative. (Lie derivative along a smooth vector field $\xi$ we denote $L_{\xi}$.) This does make sense, since:
\begin{itemize}
 \item
To use derivatives of distributions we
automatically need derivatives along vector fields.
 \item
Lie derivative preserves $p$-forms.
 \item
In case of Lie derivatives, we do not need to apply any additional geometric structure (such as connection in the case of covariant derivative).
\end{itemize}

There is an equivalent formulation to
\cite{GerTra}, given by \cite{Vickers}, which
takes the space of tensorial distributions to be ~~$D'(M)\otimes T^{m}_{n}(M)$.~~ Here $D'(M)$ is the
dual to the space of smooth, compactly supported $k$-form fields ($k$
dimensional space). Or in other words, it is the space of sections with distributional coefficients. In \cite{geometry} the authors
generalize the whole construction by taking the tensorial
distributions to be the dual to the space of compactly supported
sections of the bundle ~$E^{*}\otimes ~V\!ol^{1-q}$. Here
~$V\!ol^{1-q}$~ is a space of ~($1-q$)-densities and ~$E^{*}$~ is a
dual to a tensor bundle ~$E$~ (hence to the dual belong objects
given as ~$E\otimes ~V\!ol^{q}$). In all these formulations Lie
derivative along a smooth vector field represents the differential
operator\footnote{We can mention also another classical formulation
of distributional form fields, which comes from the old book of
deRham \cite{deRham}~(it uses the expression ``current''). It
naturally defines the space of distributions to be a dual space to
space of all compactly supported form fields.}.

Let us mention here that there is one
unsatisfactory feature of these constructions, namely that for
physical purposes we need much more to incorporate the concept
of the covariant derivative rather than the Lie derivative. There
was some work done in this direction \cite{Hartley, Marsden, Wilson}, but it
is a very basic sketch, rather than a full and satisfactory theory.
It is unclear (in the papers cited) how one can obtain for the
covariant derivative operator the expected and meaningful results
outside the class of smooth tensor fields.

\subsection{Colombeau algebra in the geometric approach}

\subsubsection{Scalar \emph{special} Colombeau algebra}

For arbitrary general manifold it is easy to find a covariant formulation of the \emph{special}
Colombeau algebra. At the end of the day you obtain a
space of $\epsilon$-sequences of functions on the general manifold $M$.
For the non-canonical case the definitions are
similar to the non-geometric formulations, the basic difference
is that Lie derivative plays here the role of the $\mathbb{R}^{n}$
partial derivative. Thus the
definition of the Colombeau algebra will be again:
\begin{equation}
\mathcal{G}(\Omega)=\mathcal{E}_{M}(\Omega)/\mathcal{N}(\Omega).
\end{equation}
$\mathcal{E}_{M}(\Omega)$ (moderate functions) are defined as
algebra of functions~ $\Omega \times (0,1]\to \mathbb{R}$, such that
are smooth on $\Omega$ (this is usually called $\mathcal{E}
(\Omega)$) and for any $K\subset\subset\Omega$ we insist that
\begin{eqnarray}
\forall k\in\mathbb{N}^{n}_{0},~\exists p\in\mathbb{N}~ \hbox{
such that } ~\forall \xi_{1}...\xi_{k}~ \hbox{ which are smooth
vector fields, }\nonumber\\  ~\sup_{x\in
K}|L_{\xi_{1}}...L_{\xi_{k}}f_{\epsilon}(x)|\leq O(\epsilon^{-p})
~\hbox{ as } ~\epsilon\to 0.~~~~~~~~~~~~~~~~~~~~~~~~~~~~~~\label{Def.}
\end{eqnarray}
$\mathcal{N}(\Omega)$ (negligible functions) we define exactly in
the same analogy to the non-geometric formulation\footnote{This
version is due to \cite{Schwarzschild}. There are also different
definitions: in \cite{geometry} the authors use instead of ``for
every number of Lie derivatives along all the possible smooth vector
fields'' the expression ``for every linear differential operator'',
but they prove that these definitions are equivalent. This is also
equivalent to the statement that in any chart holds:
$\Phi\in\mathcal{E}_{M}(\mathbb{R}^{n})$ (see
 \cite{geometry}).}:
\begin{eqnarray}
\forall K\subset\subset\Omega,~\forall k\in\mathbb{N}^{n}_{0},~\forall p\in\mathbb{N},~~\forall \xi_{1}...\xi_{k}~ \hbox{ which are smooth
vector fields, } \nonumber\\  ~\sup_{x\in
K}|L_{\xi_{1}}...L_{\xi_{k}}f_{\epsilon}(x)|\leq O(\epsilon^{p})
~\hbox{ as } ~\epsilon\to 0.~~~~~~~~~~~~~~~~~~~~~~~~~~~~~~\label{Def.}
\end{eqnarray}

\subsubsection{Tensor \emph{special} Colombeau algebra and the embedding of distributions}

After one defines the scalar special Colombeau algebra, it is easy to define the generalized Colombeau tensor
algebra as the tensor product of sections of a tensor bundle and Colombeau algebra. This
can be formulated more generally \cite{geometry} in terms of
maps from $M$ to arbitrary manifold. One can define them by changing the absolute value in the definition \eqref{Def.} to the expression ``any
Riemann measure on the target space''. Then you get the algebra \cite{geometry, generalized, Vickers}:
\begin{equation}
\mathbf{\Gamma_{C}}(X,Y)=\mathbf{\Gamma}_{M}(X,Y)/\mathbf{N}(X,Y).
\end{equation}
The tensor fields are represented when the target space is taken to
be the $TM$ manifold\footnote{This is equivalent to
~$\mathcal{G}(X)\otimes\Gamma(X,E)$ tensor valued Colombeau
generalized functions \cite{generalized}.}. It is clear that any
embedding of distributions into such algebra will be non-canonical
from various reasons. First, it is non-canonical even on
$\mathbb{R}^{n}$. Another, second reason is that this embedding will
necessarily depend on some preferred class of charts on $M$. The
embedding one defines as \cite{geometry}:

We pick an atlas, and take a smooth partition of
unity subordinate to ~$V_{\alpha_{i}}$,~ ($\Theta_{j}$,\\~
$\supp(\Theta)\subseteq V_{\alpha_{j}}$ $j\in\mathbb{N}$) and we
choose for every $j$,~ $\xi_{j}\in D(V_{\alpha_{j}})$, such that
$\xi_{j}=1$~ on ~$\supp(\xi_{j})$.~ Then we can choose in fixed charts
an $A^{\infty}(\mathbb{R}^{n})$ element $\rho$, and the embedding is given by
\begin{equation}
\sum^{\infty}_{j=1}\left(((\xi_{j}(\Theta_{j}\circ\psi^{-1}_{\alpha_{j}})u_{\alpha_{j}})*\rho_{\epsilon})\circ
\psi_{\alpha_{j}}\right)_{\epsilon},
\end{equation}
 where $\psi_{\alpha_{j}}$ is a
coordinate mapping and ~$*$~ is a convolution.

\subsubsection{The equivalence relation}

Now let us define the equivalence relation in analogy to the $\mathbb{R}^{n}$
case. Since in \cite{geometry} the strongest constraint
on the mollifier is taken, one would expect that strong
results will be obtained, but the definition is more
complicated. And in fact, standard results (such as
embedding of smooth function multiplying distribution is equivalent to
product of their embeddings) are not valid here
\cite{geometry}. That is why the stronger concept of
$k$-association is formulated. It states that $U\in\Gamma_{C}$ is $k$ associated to
function $f$, if ~~

\begin{equation}
\lim_{\epsilon\to 0} L_{\xi_{1}}...L_{\xi_{l}}
(U_{\epsilon}-f)\to 0
\end{equation}
 uniformly on compact sets for all ~$l\leq
k$. The cited paper does not contain a precise definition of $k$ equivalence
between two generalized functions, but
it can be easily derived.

\subsubsection{The older formulation of scalar \emph{full} Colombeau algebra}

If we want to get a canonical formulation, we certainly cannot
generalize it straight from the $\mathbb{R}^{n}$ case (the reason is
that the definition of the classes $A^{n}(\mathbb{R}^{n})$ is not
diffeomorphism invariant). However, there is an approach providing
us with a canonical formulation of generalized scalar fields
\cite{global}. The authors define the space $\mathcal{E}(M)$ as a
space of $C^{\infty}(M~\times ~A^{0}(M))$, where $A^{0}(M)$ is the
space of $n$-forms ($n$-dimensional space), such that ~$\int
\omega=1$. Now the authors define a smoothing kernel as $C^{\infty}$
map from
\begin{equation}
M\times I\to A^{0}(M),
\end{equation}
such that it satisfies:
\begin{itemize}
 \item[(i)] $\forall K\subset\subset M~~ \exists\epsilon_{0},~~ C>0~~ \forall p\in K~~ \forall\epsilon\leq\epsilon_{0},~~ \supp~\phi(\epsilon, p)\subseteq B_{\epsilon C}(p)$,
 \item[(ii)] $\forall K\subset\subset M, ~~\forall k,l\in\mathbb{N}_{0}, ~~\forall X_{1},\cdots X_{k},Y_{1}\cdots Y_{l}$~ smooth vector fields,\\ $\sup_{p\in K, q\in M}\left\Vert L_{Y_{1}}\cdots L_{Y_{l}}(L'_{X_{1}}+L_{X_{k}})\cdots (L'_{X_{k}}\cdots L_{X_{k}})\Phi(\epsilon,p)(q)\right\Vert = \\ = O(\epsilon^{-(n+1)})$.
\end{itemize}
Here $L'$ is defined as: ~
\begin{equation}
L'_{X}f(p,q)=L_{X}(p\to f(p,q))=\frac{d}{dt}f((Fl^{x}_{t})(p),q)|_{0}.
\end{equation}
$B_{\epsilon C}$ is a ball centered at $C$ having radius
${\epsilon}$ measured relatively to arbitrary Riemannian metric. Let
us call the class of such smoothing kernels ~$A^{0}(M)$. Then in
\cite{global} classes $A^{m}(M)$ are defined as the set of all
$\Phi\in A^{0}(M)$ such that $\forall f\in C^{\infty}(M)$ and
$\forall K\subset \subset M$ (compact subset) it holds:
\begin{equation}
\sup\left| f(p)-\int_{M}f(q)\Phi(\epsilon,p)(q)\right|=O(\epsilon^{m+1}).
\end{equation}
Moderate and the negligible functions are defined in the following way:\\
$R\in \mathcal{E}(M)$ is moderate if ~~$\forall K\subset\subset
M~~\forall k \in\mathbb{N}_{0} ~~\exists N \in \mathbb{N}~~\forall
X_{1},....X_{k}$ ($X_{1},....X_{k}$ are smooth vector fields) and
~$\forall \Phi\in A^{0}(M)$~ one has:
\begin{equation}
\sup_{p\in
K}~\left\Vert L_{X_{1}}.....L_{X_{k}}(R(\Phi(\epsilon,p),p))\right\Vert=O(\epsilon^{-N}).
\end{equation}
$R\in\mathcal{E}(M)$ is negligible if ~~$\forall K\subset\subset
M,~\forall k,l\in\mathbb{N}_{0}~~\exists m\in\mathbb{N}~~\forall
X_{1},...X_{k}$ ($X_{1},...X_{k}$ are again smooth vector fields)
and ~$\forall\Phi\in A_{m}(M)$~ one has:
\begin{equation}
\sup_{p\in
K}~\left\Vert L_{X_{1}}...L_{X_{k}}(R(\Phi(\epsilon,p),p)\right\Vert=O(\epsilon^{l}).
\end{equation}
Now we can define the Colombeau algebra in the usual way as a factor
algebra of moderate functions over negligible functions. Scalar distributions, defined
as dual to $n$-forms, can be embedded into such algebra in a complete analogy
to the canonical $\mathbb{R}^{n}$ formulation. Also association is in this case
defined in the ``usual'' way (integral with compactly supported
smooth $n$-form field) and has for multiplication the usual properties.
However any attempt to get a straightforward generalization from scalars to tensors
brings immediate problems, since the embedding does not commute
with the action of diffeomorphisms. This problem was finally resolved in
\cite{CanTen}.

\subsubsection{Tensor \emph{full} Colombeau algebra and the embedding of distributions}

The authors of reference \cite{CanTen} realized that
diffeomorphism invariance can be achieved by adding some background structure
defining how tensors transport from point to point, hence a
transport operator. Colombeau  $(m,n)$ rank tensors are then taken
from the class of smooth maps ~~$C^{\infty}(\omega, q, B)$~~having values in
~$(T^{m}_{n})_{q}M$, ~~ where ~$\omega\in A^{0}(M)$, $q\in M$  ~and~
$B$~ is from the class of compactly supported transport operators.
After defining how Lie derivative acts on such objects and
the concept of the ``core'' of a transport operator, the authors of reference \cite{CanTen} define (in a
slightly complicated analogy to the previous case) the moderate and the
negligible tensor fields. Then by usual factorization they obtain the
canonical version of the generalized tensor fields (for more details
see \cite{CanTen}). The canonical embedding of tensorial
distributions is the following:
The smooth tensorial objects are embedded as

\begin{equation}
\tilde t(p,\omega,B)=\int t(q)B(p,q)\omega(q) dq
\end{equation}
where as expected $\omega\in A^{0}(M)$, $t$ is the smooth tensor field and
$B$ is the transport operator. Then the arbitrary tensorial distribution $s$ is embedded (to
$\tilde s$) by the condition

\begin{equation}
\tilde s(\omega, p, B)\cdot t(p)=\left(s, B(p,.)\cdot
t(p)\otimes\omega(.)\right),
\end{equation}
where on the left side we are contracting the embedded object with a
smooth tensor field $t$, and on the right side we are applying the
given tensorial distribution $s$ in the variable assigned by the
dot. It is shown that this embedding fulfills all
the important properties, such as commuting with the Lie
derivative operator \cite{CanTen}. All the other results related to equivalence relation
(etc.) are obtained in complete analogy to the previous cases.

\subsubsection{The generalized geometry (in \emph{special} Colombeau algebras)}

In \cite{geometry, Connection, generalized} the authors generalized all the basic geometric structures, like connection, covariant derivative, curvature, or geodesics into the geometric formulation of the special Colombeau algebra. That means they defined the whole generalized geometry.

\subsubsection{Why is this somewhat unsatisfying?}

 However in our view, the crucial part is missing. What we would like to see is an intuitive and clear definition of the covariant derivative operator acting on the distributional objects in the canonical Colombeau algebra formulation, on one hand reproducing all the classical results, and on the other hand extending them in the same natural way as in the classical distributional theory with the classical derivative operator. Whether there is any way to achieve this goal by the concept of generalized covariant derivative acting on the generalized tensor fields, as defined in \cite{generalized}, is unclear. Particularly it is not clear whether such generalized geometry can be formulated also within the canonical Colombeau algebra approach. There exist definitions of the covariant derivative operator within the distributional tensorial framework \cite{Hartley, Marsden}. (The reference \cite{Hartley} gives particularly nice application of such distributional tensor theory to signature changing spacetimes.) But these approaches are still ``classical'', in the sense, that they do not fully involve the operation of the tensor product of distributional tensors.
There cannot be any hope of
finding a more appropriate, generalized formulation of classical physics
without finding such a clear and intuitive definition of both, the covariant derivative and the tensor product. All that can be in this situation achieved is
to use these constructions to solve some specific problems within
the area of physics. But as we see, the more ambitious goal can be
very naturally achieved by our own construction, which follows after
this overview.

\subsection{Practical application of the standard results}

Now we will briefly review various applications of
the Coulombeau theory presented before.

\subsubsection{Classical shock waves}
The first application we will mention is the non-general-relativistic one. In
\cite{shock} the authors provide us with weak solutions of
nonlinear partial differential equations (using Colombeau algebra)
representing shock waves. They use a special version of the
Colombeau algebra, and specifically the relation
\begin{equation}
H^{n}H'\approx \frac{1}{n+1}~H',
\end{equation}
(which is related to mollifiers from
$A^{0}(\mathbb{R}^{n})$ class). The more general analysis related to the existence and uniqueness of
weak solutions of nonlinear partial differential
equations can be found in \cite{Oberguggenberger}.

\subsubsection{Black hole ``distributional'' spacetimes}
In general relativity there are many results obtained by the use of Colombeau
algebras. First, we will focus on the distributional Schwarzschild geometry, which is
analysed for example in \cite{Schwarzschild}. The authors of \cite{Schwarzschild} start to
work in Schwarzschild coordinates using the special Colombeau algebra
and $A^{\infty}(\mathbb{R}^{n})$ classes of mollifiers. They obtain
the delta-functional results (as expected) for the Einstein tensor, and
hence also for the stress-energy tensor. But in Schwarz\-schild coordinates there
are serious problems with the embedding of the distributional tensors, since these coordinates do not
contain the 0 point. As a result, if one looks for smooth embeddings, one does
not obtain an inverse element in the Colombeau algebra in the
neighbourhood of 0 for values of $\epsilon$ close to 0. (Although there is no problem, if we
require that the inverse relation should apply only in the sense of
equivalence.) Progress can be made by
turning to Eddington-Finkelstein coordinates \cite{Schwarzschild}. The metric is obtained in Kerr-Schild form, in which one is able to compute $R^a{}_b,
G^a{}_b$ and hence $T^a{}_b$ as delta-functional objects (which is expected). (The (1,1) form of the field equations is used since the metric dependence has a relatively simple form in Eddington-Finkelstein coordinates.) This
result does not depend on the mollifier (see also \cite{Vickers}),
but one misses the analysis of the relation between different embeddings given by the different
coordinate systems \footnote{It seems that
the authors use relations between $R^\mu{}_\nu$ and components of
$g^{\mu\nu}$ obtained in Eddington-Finkelstein coordinates by using algebraic tensor
computations. Then it is not obvious, whether these results can be
obtained by computation in the Colombeau algebra using the $\approx$
relation, since in such case some simple tricks (such as substitution) cannot in general be used.}.
Even in the case of Kerr geometry there is a computation
of $\sqrt{(g')}R^\mu{}_\nu$ (where $g^{ab}=g'^{ab}+fk^{a}k^{b}$)
given by Balasin, but this is
mollifier dependent \cite{Kerr, Vickers}. Here the
coordinate dependence of the results is even more
unclear.

\subsubsection{Aichelburg metric}
There exists an ultrarelativistic weak limit of the Schwarzschild
metric. It is taken in Eddington-Finkelstein coordinates $u=t+r^{*}$,
$v=t-r^{*}$, (where $r^{*}$ is tortoise coordinate) by taking the boost in the weak $v\to c$ limit.
We obtain the ``delta functional'' Aichelburg metric. Reference \cite{geodesic} provides a computation of geodesics in such a geometry. The authors of \cite{geodesic} take the special Colombeau algebra (and take
$A^{0}(\mathbb{R}^{n})$ as their class of mollifiers), and they prove
that geodesics are given by the refracted lines. The results are mollifier
independent. This is again expected. Moreover, what seems to be really
interesting is that there is a continuous metric which is connected with
the Aichelburg metric by a generalized coordinate transformation
\cite{Penrose, Vickers}.

\subsubsection{Conical spacetimes}
The other case we
want to mention are conical spacetimes. One of the papers where the conical spacetimes are analysed
is an old paper of Geroch and Traschen \cite{GerTra}. In \cite{GerTra} it is shown that conical spacetimes can not be analysed through the concept of
$gt$-metric. These are metrics which provide us with a distributional Ricci
tensor in a very naive sense. The multiplication is given just by a simple product
of functions defining the regular distributions. A calculation of the stress energy tensor was given by Clarke, Vickers and Wilson \cite{Clarke}, but this is mollifier dependent (although it is
coordinate independent \cite{Vickers}).

\section{How does our approach relate to current theory?}

\paragraph{General summary}
How does our own approach (to be described in detail in the next section) relate to all what has presently been achieved in the Colombeau theory? We can summarize what we will do in three following points:

\begin{itemize}
\item
We will define tensorial distributional objects, and the basic related operations (especially the \emph{covariant derivative}). The definition directly follows our physical intuition (there is a unique way of constructing it). This generalizes the Schwartz $\mathbb{R}^{n}$ distribution theory.

\item
We will formulate the Colombeau equivalence relation in our approach and obtain all the usual equivalence results from the Colombeau theory.

\item
We will prove that the classical results for the covariant derivative operator (as known within differential geometry) significantly generalize in our approach.
\end{itemize}
The most important point is that our approach is fully based on the Colombeau equivalence relation translated to our language. That means we take and use only this particular feature of the Colombeau theory and completely avoid the Colombe\-au algebra construction.

\paragraph{The advantages compared to usual Colombeau theory}
What are the advantages of this approach comparing to Colombeau theory?
\begin{itemize}
\item
First, by avoiding the algebra factorization (which is how Colombeau algebra is constructed) we fulfill the physical intuitiveness condition of the language used.
\item
Second, we can naturally and easily generalize the concept of the covariant derivative in our formalism, which has not been completely satisfactorily achieved by the Colombeau algebra approach\footnote{As previously mentioned, the current literature lists some ambiguous attempts to incorporate the covariant derivative, but no satisfactory theory.}. This must be taken as an absolutely necessary condition that any generalization of a fundamental physical language must fulfill.
\item
It is specifically worth discussing the third advantage:  Why is the classical approach so focused on Colombeau
algebras? The answer is simple: We want to get an algebra of
$C^{\infty}(M)$ functions as a subalgebra of our algebra. (This is
why we need to factorize by negligible functions, and we need to get
the largest space where they form an ideal, which is the space of
moderate functions.) But there is one strange thing: all our efforts
are aimed at reaching the goal of getting a more rich space than is provided by the space
of smooth functions. But there is no way of getting a larger
differential subalgebra than the algebra of smooth functions, as
is shown by Schwartz impossibility result. That is why we use only
the equivalence relation instead of straight equality. But then the
question remains: Why one should still prefer smooth
functions? Is not the key part of all the theory the equivalence
relation? So why is it that we are not satisfied with the way the equivalence relation recovers multiplication of the smooth objects (we require something stronger), but we are satisfied with the way it recovers multiplication
within the larger class? Unlike the Colombeau algebra based approach, we are simply taking seriously the idea that one should treat all the objects in an equal way. This means we do not see any reason to try to achieve ``something more'' with smooth objects than we do with objects outside this class. And the fact that we treat all the objects in the same way provides the third advantage of our theory; it makes the theory much more natural than the Colombeau algebra approach.
\item
The fourth advantage is that it naturally works with the much more general (and for the language of distributions natural) concept of piecewise smooth manifolds, so the generalization of the physics language into such conceptual framework will give it much higher symmetry.
\item
The fifth and last, ``small'' advantage comes from the fact that by avoiding the Colombeau algebra construction we automatically remove the problem of how to canonically embed arbitrary tensorial distributions into the algebra. But one has to acknowledge that this problem was already solved also within the Colombeau theory approach \cite{CanTen}.
\end{itemize}

\paragraph{The disadvantages compared to usual Colombeau theory}
What are the disadvantages of this approach compared to Colombeau theory? A conservative person might be not satisfied with the fact that we do not have a smooth tensor algebra as a subalgebra of our algebra. This means that the classical smooth tensorial fields have to be considered to be solutions of equivalence relations only (as opposed to the classical, ``stronger'' view to take them as solutions of the equations). But in my view, this is not a problem at all. The
equivalence relations contain all the classical smooth tensorial solutions
(equivalently smooth tensor fields), for the smooth initial value
problem. So for the smooth initial value problem the equivalence relations
reduce to classical equations. We will suggest how to extend the initial
value problem for larger classes of distributions and
show that it has unique solutions. It is true that the equivalence relations
might have also many other (generally non-linear) solutions in the
space constructed,\footnote{By ``solution'' we mean here any
object fulfilling the particular equivalence relations.} but the
situation that there exist many physically meaningless solutions is
for physicists certainly nothing new. Our previous considerations
suggest that we shall look only for distributional solutions (they
are the ones having physical relevance) where the solution is
provided to be unique (up to initial values obviously), if it exists.
(This will also recover the common, but also many ``less'' common
physical results \cite{Vickers}.)

\section{New approach}\label{NewApp}

\subsection{The basic concepts/definitions}\label{basicconcepts}

Before saying anything about distributional
tensor fields, we have to define the basic concepts which will be
used in all the following mathematical constructions. This task is
dealt with in this section, so this is the part crucially important for
understanding all the subsequent theory. An attentive reader, having read through the introduction and abstract will understand
why we are particularly interested in defining these concepts.

\subsubsection{Definition of (M,$\mathcal{A}$)}

\theoremstyle{definition}
\newtheorem{defn}{Definition}[section]
\begin{defn}
 By piecewise smooth function we mean a function from an open set
$\Omega_{1}(\subseteq\mathbb{R}^{4})\to\mathbb{R}^{m}$ such, that there exists
an open set $\Omega_{2}$ (in the usual ``open ball'' topology on
$\mathbb{R}^{4}$) on which this function is smooth and
$\Omega_{2}=\Omega_{1}\setminus \Omega'$ (where $\Omega'$ has a Lebesgue measure 0).
\end{defn}

Take $M$ as a 4D paracompact\footnote{We will use specifically
4-dimensional manifolds, but one can immediately generalize all the following
constructions for $n$-dimensional manifolds.}, Hausdorff locally
Euclidean space, on which there exists a smooth atlas
$\mathcal{\mathcal{\mathcal{\mathcal{\mathcal{S}}}}}$. Hence
$(M,\mathcal{\mathcal{\mathcal{\mathcal{S}}}})$ is a smooth
manifold. Now take a ordered couple $(M,\mathcal{A})$, where
$\mathcal{A}$ is the maximal atlas, where all the maps are
connected by piecewise smooth transformations such that:
\begin{itemize}
\item
the transformations and their inverses have on every compact subset of $\mathbb{R}^{4}$ all the first derivatives (on the domains
where they exist) bounded\\ (hence Jacobians, inverse
Jacobians are on every compact set bounded),
\item
it contains at least one maximal smooth subatlas $\mathcal{\mathcal{\mathcal{\mathcal{S}}}}\subseteq\mathcal{A}$,\\ (coordinate transformations between maps are smooth
there).
\end{itemize}
\medskip
\paragraph{Notation.} The following notation will be used:
\begin{itemize}
\item
By the letter $\mathcal{\mathcal{\mathcal{\mathcal{\mathcal{S}}}}}$ we will always mean some maximal smooth subatlas of $\mathcal{A}$.
\item
Every subset of $M$ on
which there exists a chart from our atlas $\mathcal{A}$, we call $\Omega_{Ch}$.
An arbitrary chart on $\Omega_{Ch}$ from our atlas $\mathcal{A}$ is denoted
$Ch(\Omega_{Ch})$.
\end{itemize}

\paragraph{Notation.}
 Take some set $\Omega_{Ch}$. Take some open subset of that set
$\Omega'\subset\Omega_{Ch}$. Then $ Ch(\Omega_{Ch})_{|\Omega'}$  is defined simply as
$Ch(\Omega')$, which is obtained from $Ch(\Omega_{Ch})$ by limiting the domain to $\Omega'$.

\subsubsection{Definition of ``continuous to the maximal possible degree''}
\theoremstyle{definition}
\newtheorem{defn3.001}[defn]{Definition}
\begin{defn3.001}
We call a function on $M$ continuous to the maximal possible degree, if on arbitrary $\Omega$ of Lebesgue measure 0 it is only in such cases:\footnote{The expression
``Lebesgue measure 0 set'' will have in this paper extended meaning. It refers to such subsets of a general manifold $M$ that they have in arbitrary chart Lebesgue measure 0.}

\begin{itemize}
\item[a)]
either undefined,
\item[b)]
or defined and discontinuous,
\end{itemize}
 in which there does not exist a way of turning it into a function continuous on $\Omega$ by
\begin{itemize}
\item
in the case of $a)$ extending its domain by the $\Omega$ set,
\item
in the case of $b)$ re-defining it on $\Omega$.
\end{itemize}
\end{defn3.001}

\subsubsection{Jacobians and algebraic operations with Jacobians} Now it is
obvious that since transformations between maps do not have to be
everywhere once differentiable, the Jacobian and inverse Jacobian may always be undefined on a set having Lebesgue measure 0. Now if we understand product in
the sense of a limit, then the relation
\begin{equation}
J^{\mu}_{\alpha}~(J^{-1})^{\alpha}_{\nu}=\delta^{\mu}_{\nu},\label{maxcon}
\end{equation}
for example, might hold even at the points where both Jacobian and inverse
Jacobian are undefined. This generally means the following: any algebraic operation with tensor fields is understood in such way, that in every chart it gives sets of functions continuous to a maximal possible degree. From this follows that the matrix product (\ref{maxcon})
must be, for $\mu=\nu$, equal to 1 and, for $\mu\neq\nu$, ~0.

\subsubsection{Tensor fields on $M$} We understand the tensor field on
$M$ to be an object which is:
\begin{itemize}
\item
Defined relative to the 1-differentiable subatlas of $\mathcal{A}$ everywhere except for a set
having Lebesgue measure 0 (this set is a function of the given 1-differentiable subatlas).
\item
In every chart from $\mathcal{A}$ it is given by
functions continuous to a maximal possible degree.
\item
It transforms $\forall\Omega_{Ch}$ between charts
$Ch_{1}(\Omega_{Ch}), ~Ch_{2}(\Omega_{Ch})\in\mathcal{A}$ in the tensorial way

\begin{equation}
T^{\mu...}_{\nu...~Ch_{2}}=J^{\mu}_{\alpha}~(J^{-1})^{\beta}_{\nu}\cdots
~T^{\alpha...}_{\beta...~Ch_{1}}~~~~~~\hbox{a.e.},\footnote{This expression means ``almost everywhere'', that is, ``everywhere apart from a set having Lebesgue measure 0''.}
\end{equation}

 where $J^{\mu}_{\alpha}$ is the Jacobian of the coordinate transformation from $Ch_{1}(\Omega_{Ch})$ to $Ch_{2}(\Omega_{Ch})$, and $T^{\mu...}_{\nu...~Ch_{1}}, T^{\mu...}_{\nu...~Ch_{2}}$ are tensor field components in charts $Ch_{1}$, $Ch_{2}$.
As we already mentioned: If $T^{\mu...}_{\nu...~Ch_{1}}$ is at some
given point undefined, so in some chart $Ch_{1}(\Omega_{Ch})$ the tensor components
do not have a defined limit,
this limit can still exist in $Ch_{2}(\Omega_{Ch})$, since Jacobians
and inverse Jacobians of the transformation from
$Ch_{1}(\Omega_{Ch})$ to $Ch_{2}(\Omega_{Ch})$ might be undefined at
that point as well. This limit then defines
$T^{\mu...}_{\nu...~Ch_{2}}$ at the given point.
\end{itemize}

\subsubsection{Important classes of test objects}

\paragraph{Notation.} The following notation will be used:
\begin{itemize}
\item
We denote by $C^{P}(M)$ the class of 4-form fields on $M$, such that they are
compactly supported and their support lies within some $\Omega_{Ch}$. For such 4-form fields we will generally use the symbol ~$\omega$.
\item
For the scalar density related to $\omega\in C^{P}(M)$ we use always the symbol ~$\omega'$.
\item
By $C^{P}(\Omega_{Ch})$ we mean a subclass of $C^{P}(M)$, given by 4-form fields having support inside $\Omega_{Ch}$. Note that only the
$C^{P}(\Omega_{Ch})$ subclasses form linear spaces.
\item
Take such maximal atlas $\mathcal{\mathcal{\tilde S}}$, ~($\exists
\mathcal{\mathcal{\mathcal{\mathcal{S}}}}\subset
\mathcal{\mathcal{\tilde S}}\subset\mathcal{A}$),~ that there exist
4-forms from $C^{P}(M)$, such that they are given in this atlas by
everywhere smooth scalar density ~$\omega'$. (``Maximal'' here means
that these 4-forms have in every chart outside this atlas non-smooth
scalar densities.) By $C^{P}_{S (\mathcal{\mathcal{\tilde S}})}(M)$
we mean a class of all such elements from $C^{P}(M)$, that they have
everywhere smooth scalar density in ~$\mathcal{\tilde S}$.
\item
The letter $\mathcal{\mathcal{\tilde S}}$ will from now on be reserved for maximal atlases defining $C^{P}_{S (\mathcal{\mathcal{\tilde S}})}(M)$ classes.
\item
$C^{P}_{S}(M)$ is defined as: $C^{P}_{S}(M)\equiv \cup_{\mathcal{\mathcal{\tilde S}}}C^{P}_{S (\mathcal{\mathcal{\tilde S}})}(M)$.
\item
$C^{P}_{S (\mathcal{\mathcal{\tilde S}})}(\Omega_{Ch})$ ~(or $C^{P}_{S}(\Omega_{Ch})$)
means $C^{P}_{S (\mathcal{\mathcal{\tilde S}})}(M)$  ~(or $C^{P}_{S}(M)$) element having
support inside the given $\Omega_{Ch}$.
\end{itemize}

\subsubsection{Topology on $C^{P}_{S (\mathcal{\mathcal{\tilde S}})}(\Omega_{Ch})$}
Consider the following topology on each $C^{P}_{S (\mathcal{\tilde
S)}}(\Omega_{Ch})$: A sequence from $\omega_{n}\in C^{P}_{S
(\mathcal{\mathcal{\tilde S}})}(\Omega_{Ch})$ converges to an
element ~$\omega$~ from that set if all the supports of
~$\omega_{n}$~ lie in a single compact set, and in any chart
$Ch(\Omega_{Ch})\in \mathcal{\mathcal{\tilde S}}$, for arbitrary
$k$ it is true that $\frac{\partial^{k}\omega'_{n}(x^{i})}{\partial
x^{l_{1}}..\partial x^{l_{k}}}$ converges uniformly to
$\frac{\partial^{k}\omega'(x^{i})}{\partial x^{l_{1}}..\partial
x^{l_{k}}}$.

\subsection{Scalars}

This section deals with the definition of scalar distributions as the easiest
particular example of a generalized tensor field. The explanatory reasons are the main ones
why we deal with scalars separately, instead of taking more
``logical'', straightforward way to tensor fields of arbitrary rank.

\subsubsection{Definition of $D'(M)$, and hence of \emph{linear generalized scalar fields}}

\theoremstyle{definition}
\newtheorem{defn4.99}{Definition}[section]
\begin{defn4.99}
We say that $B$, being a function that maps some subclass of $C^{P}(M)$ to $\mathbb{R}^{n}$, is linear, if the following holds: Take such $\omega_{1}$ and $\omega_{2}$ from the domain of $B$, that they belong to the same class $C^{P}(\Omega_{Ch})$.
Whenever the domain of $B$ contains also their linear combination
$\lambda_{1}\omega_{1}+\lambda_{2}\omega_{2}$, where ~$\lambda_{1},\lambda_{2}\in\mathbb{R}$,~ then ~$B(\lambda_{1}\omega_{1}+\lambda_{2}\omega_{2})=\lambda_{1}B(\omega_{1})+\lambda_{2}B(\omega_{2})$.\label{linear}
\end{defn4.99}

\theoremstyle{definition}
\newtheorem{defn4.999}[defn4.99]{Definition}

\begin{defn4.999}
 Now take the space of linear maps $F\to\mathbb{R}$, where $F$ is such set that $\exists\mathcal{\mathcal{\tilde S}}$, such that
$C^{P}_{S (\mathcal{\mathcal{\tilde S}})}(M)\subseteq F\subseteq
C^{P}(M)$. These linear maps are also required to be for every
$\Omega_{Ch}$ on $C^{P}_{S \mathcal{(\tilde S)}}(\Omega_{Ch})$
continuous, (relative to the topology taken in section
\ref{basicconcepts}). Call set of such maps $D'(M)$, or in words,
the set of \emph{linear generalized scalar fields}.
\end{defn4.999}

%\bigskip

\subsubsection{Important subclasses of $D'(M)$}

\paragraph{Notation.} The following notation will be used:
\begin{itemize}
\item
Now take a subset of $D'(M)$ given by regular distributions
defined as integrals of piecewise continuous
functions (everywhere on $C^{P}(M)$, where it converges). We denote it by $D'_{E}(M)$.\footnote{Actually, it holds that if and only if the function is integrable in every chart on every compact set in $\mathbb{R}^n$, then this function defines a regular $D'(M)$ distribution and is defined at least on the whole $C^{P}_{S}(M)$ class.}
\item
Take such subset of $D'_{E}(M)$ that there  $\exists \mathcal{\mathcal{\mathcal{\mathcal{S}}}}$
in which the function under the integral is smooth. Call this class
$D'_{S}(M)$.
\item
 Now~ take ~subsets ~of ~$D'(M)$~ such ~that~ they~ have
some ~common ~set\\
$\cup_{n}C^{P}_{S (\mathcal{\mathcal{\tilde S}}_{n})}(M)$
belonging to their domains and are ~~ $\forall\Omega_{Ch}$,
~$\forall n$~ continuous on $C^{P}_{S (\mathcal{\mathcal{\tilde
S}}_{n})}(\Omega_{Ch})$. Denote such subsets by
$D'_{(\cup_{n}\mathcal{\tilde S}_{n})}(M)$. Obviously ~$T\in
D'_{(\cup_{n}\mathcal{\mathcal{\tilde S}}_{n})}(M)$ means
~$T\in\cap_{n}D'_{(\mathcal{\mathcal{\tilde S}}_{n})}(M)$.
\item
By $D'_{(\cup_{n}\mathcal{\mathcal{\tilde S}}_{n} o)}(M)$ we mean
objects such that they belong to $D'_{(\cup_{n}\mathcal{\mathcal{\tilde S}}_{n})}(M)$ and their full domain is given as ~$\cup_{n}C^{P}_{S (\mathcal{\mathcal{\tilde S}}_{n})}(M)$.
\item
If we use
the notation $D'_{E (\cup_{n}\mathcal{\mathcal{\tilde S}}_{n} o)}(M)$, we mean objects defined by integrals of piecewise continuous functions, with their domain being the class $\cup_{n}C^{P}_{S (\mathcal{\mathcal{\tilde S}}_{n})}(M)$.
\item
By using $D'_{S(\cup_{n}\mathcal{\mathcal{\tilde S}}_{n} o)}(M)$ we automatically mean subclass of $D'_{E(\cup_{n}\mathcal{\mathcal{\tilde S}}_{n} o)}(M)$, such that it is given by
an integral of a smooth function in some smooth subatlas
$\mathcal{\mathcal{\mathcal{\mathcal{S}}}}\subseteq\cup_{n}\mathcal{\mathcal{\tilde S}}_{n} $.
\end{itemize}

\subsubsection{$D'_{A}(M)$, hence \emph{generalized scalar fields}}

\paragraph{Notation.}
 Let us for any arbitrary set $D'_{(\mathcal{\mathcal{\tilde S}})}(M)$\footnote{The union is here trivial, it just means one element $\mathcal{\mathcal{\tilde S}}$.} construct, by the use of pointwise multiplication of its elements,
an algebra. Another way to describe the algebra is that it is a set of multivariable arbitrary degree polynomials, where different variables represent different elements of $D'_{(\mathcal{\mathcal{\tilde S}})}(M)$. Call it $D'_{(\mathcal{\mathcal{\tilde S}}) A}(M)$.

By pointwise multiplication of linear generalized scalar fields ~~~$B_{1},~~B_{2}~~\in \\ D'_{(\mathcal{\mathcal{\tilde S}})}(M)$~ we mean a mapping from $\omega$ into product of the images (real numbers) of the $B_{1},B_{2}$ mappings:
~$(B_{1}\cdot B_{2})(\omega)\equiv B_{1}(\omega)\cdot B_{2}(\omega)$.~ The domain on which the
product (and the linear combination as well) is defined is an
intersection of domains of $B_{1}$ and $B_{2}$ (trivially always nonempty,
containing $C^{P}_{S (\mathcal{\mathcal{\tilde S}})}(M)$ at least).
Note also that the resulting arbitrary element of $D'_{(\mathcal{\mathcal{\tilde S}}) A}(M)$ has in general all the properties defining
$D'(M)$ objects, except that of being necessarily linear.
\theoremstyle{definition}
\newtheorem{defn5}[defn4.99]{Definition}
\begin{defn5}
 The set of objects obtained by the union~ $\cup_{\mathcal{\mathcal{\tilde S}}}D'_{(\mathcal{\mathcal{\tilde S}}) A}(M)$~ we denote $D'_{A}(M)$, and call
them \emph{generalized scalar fields} (GSF).
\end{defn5}

\subsubsection{Topology on $D'_{(\mathcal{\mathcal{\tilde S}} o)}(M)$}

If we take objects from $D'_{(\mathcal{\mathcal{\tilde S}} o)}(M)$ and we take a weak
($\sigma -$) topology on that
set, we know that any object is a limit of some sequence from $D'_{S
(\mathcal{\mathcal{\tilde S}} o)}(M)$ objects from that set (they form a dense subset of
that set). That is known from the classical theory. Such a
space is complete.

\subsection{Generalized tensor fields}\label{GTFsection}

\emph{This section is of crucial importance.} It provides us with
definitions of all the basic objects we are interested in, the
generalized tensor fields and all their subclasses of special
importance as well.

\subsubsection{The class $D'^{m}_{n}(M)$ of \emph{linear generalized tensor fields}}

First let us clearly state how to interpret the $J^{\mu}_{\nu}$ Jacobian in all the following definitions.~ It is a matrix of
piecewise smooth functions  $\Omega_{1}\setminus \Omega_{2}\to\mathbb{R}$,~ $\Omega_{1}$ being
a open subset of $\mathbb{R}^{4}$ and $\Omega_{2}$ having Lebesgue measure 0. Let it represent transformations from $Ch_{1}(\Omega_{Ch})$ to $Ch_{2}(\Omega_{Ch})$.
We can map the Jacobian by the inverse of the $Ch_{1}(\Omega_{Ch})$ coordinate mapping
to $\Omega_{Ch}$ and it will become a matrix of functions
$\Omega_{Ch}\setminus \Omega'\to\mathbb{R}$, ~~$\Omega'$ having Lebesgue measure 0.
~The object $J^{\mu}_{\nu}\cdot\omega$~ is then understood as a matrix of 4-forms from $C^{P}(\Omega_{Ch})$, which also means that outside $\Omega_{Ch}$ we trivially define them to be 0.

\theoremstyle{definition}
\newtheorem{defn6}{Definition}[section]
\begin{defn6}
Take some set $F$, ~$F_{1}\subseteq F\subseteq F_{2}$, where $F_{1}$
us such set that ~~$\exists\mathcal{\mathcal{\tilde S}}$~ and
~~$\exists \mathcal{\mathcal{\mathcal{\mathcal{S}}}}\subseteq
\mathcal{\mathcal{\tilde S}}$~~
($\mathcal{\mathcal{\mathcal{\mathcal{\mathcal{S}}}}}$ is some
maximal smooth atlas) defining it as:
\begin{equation*}
F_{1}\equiv \cup_{\Omega_{Ch}}\left\{
\big(\omega,Ch(\Omega_{Ch})\big):~\omega\in C^{P}_{S
(\mathcal{\tilde S})}(\Omega_{Ch}),~
Ch(\Omega_{Ch})\in\mathcal{\mathcal{\mathcal{\mathcal{\mathcal{S}}}}}\right\}.
\end{equation*}
$F_{2}$ is defined as:
\begin{equation*}
F_{2}\equiv \cup_{\Omega_{Ch}}\left\{
\big(\omega,Ch(\Omega_{Ch})\big):~\omega\in C^{P}(\Omega_{Ch}),~
Ch(\Omega_{Ch})\in\mathcal{A})\right\}.
\end{equation*}
By a ~$D'^{m}_{n}(M)$~ object, the \emph{linear generalized tensor
field}, we mean a linear mapping from ~$F\to\mathbb{R}^{4^{m+n}}$
for which the following holds: \footnote{We could also choose for
our basic objects maps taking ordered couples from
~$C^{P}(\Omega_{Ch})\times Ch(\Omega'_{Ch})$,
~($\Omega_{Ch}\neq\Omega'_{Ch}$). The linearity condition then
automatically determines their values, since for $\omega\in
C^{P}(\Omega_{Ch})$, whenever it holds that ~$\Omega'_{Ch}\cap
\hbox{supp}(\omega)=\{0\}$, they must automatically give ~0~ for any
chart argument. Hence these two definitions are trivially connected
and choice between them is just purely formal (only a matter of
``taste'').}

\begin{itemize}
\item
$\forall\Omega_{Ch}$~ it is ~$\forall Ch_{k}(\Omega_{Ch})\in
\mathcal{\mathcal{\mathcal{\mathcal{S}}}}\subset\mathcal{\mathcal{\tilde
S}}$~ continuous on the class $C^{P}_{S (\mathcal{\tilde
S})}(\Omega_{Ch})$. (Both $\mathcal{S}$, $\mathcal{\tilde S}$ are from the definition of $F_{1}$.)
\item
This map also ~~$\forall\Omega_{Ch}$ ~~transforms between two charts
from its domain,\\~ $Ch_{1}(\Omega_{Ch})$, ~$Ch_{2}(\Omega_{Ch})$, as:

\begin{equation*}
 T^{\mu...\gamma}_{\nu...\delta}(Ch_{1},\omega)=T^{\alpha...\lambda}_{\beta...\rho}\left(Ch_{2},J^{\mu}_{\alpha}...J^{\gamma}_{\lambda}(J^{-1})^{\beta}_{\nu}....(J^{-1})^{\rho}_{\delta}\cdot\omega\right)~.
\end{equation*}

\item
The following consistency condition holds: If
~$\Omega'_{Ch}\subset\Omega_{Ch}$, then
~$T^{\mu...\alpha}_{\nu...\delta}$~  gives on ~$\omega~\times
~Ch(\Omega_{Ch})_{|\Omega'_{Ch}}$, ~$\omega\in
C^{P}(\Omega'_{Ch})$ the same results\footnote{By the ``same
results'' we mean that they are defined on the same domains, and by
the same values.} as on ~$\omega~\times ~Ch(\Omega_{Ch})$.
\end{itemize}

\end{defn6}

We can formally extend this notation
also for the case $m=n=0$. This means scalars, exactly as
defined before. So from now on $m$, $n$ take also the value 0, which means the theory in the following sections holds also for the scalar objects.

\subsubsection{Important subclasses of $D'^{m}_{n}(M)$}

\paragraph{Notation.} The following notation will be used:
\begin{itemize}
\item
By a complete analogy to scalars we define classes
~$D'^{m}_{n E}(M)$:
On arbitrary $\Omega_{Ch}$, being fixed in
arbitrary chart ~$Ch_{1}(\Omega_{Ch})\in\mathcal{A}$ we can express it in
another arbitrary chart ~$Ch_{2}(\Omega_{Ch})\in\mathcal{A}$, as an integral from a
multi-index matrix of piecewise continuous
functions on such subset of ~$C^{P}(M)$, on which the integral is
convergent\footnote{Actually we will use the
expression ``multi-index matrix'' also later in the text and it just means
specifically ordered set of functions.}.
\item
Analogously the class $D'^{m}_{n S}(M)\subset D'^{m}_{n E}(M)$ ~is defined by objects which can, for some maximal
smooth atlas $\mathcal{\mathcal{\mathcal{\mathcal{\mathcal{S}}}}}$, in arbitrary charts \footnote{The first chart is
an argument of this generalized tensor field and the second chart is
the one in which we express the given integral.}
~$Ch_{1}(\Omega_{Ch})$, ~$Ch_{2}(\Omega_{Ch})\in\mathcal{\mathcal{\mathcal{\mathcal{\mathcal{S}}}}}$, ~be expressed
by an integral from a multi-index matrix of smooth functions.
\item
$D'^{m}_{n (\cup_{l}At(\mathcal{\mathcal{\tilde S}}_{l}))}(M)$~ means a class of objects being for every $\Omega_{Ch}$ in every $Ch(\Omega_{Ch})\in At(\mathcal{\mathcal{\tilde S}}_{l})$ continuous on $C^{P}_{S(\mathcal{\mathcal{\tilde S}}_{l})}(\Omega_{Ch})$. (Here $At(\mathcal{\tilde S})$ stands for a map from atlases $\mathcal{\tilde S}$ to some subatlases of $\mathcal{A}$.)
\item
$D'^{m}_{n (\cup_{l}At(\mathcal{\mathcal{\tilde S}}_{l})o)}(M)$ means a class of objects from $D'^{m}_{n (\cup_{l}At(\mathcal{\mathcal{\tilde S}}_{l}))}(M)$ having as their domain the union
\[\cup_{\Omega_{Ch}}\cup_{l}\left\{ (\omega,Ch(\Omega_{Ch})):~\omega\in C^{P}_{S (\mathcal{\mathcal{\tilde S}}_{l})}(\Omega_{Ch}),~
Ch(\Omega_{Ch})\in At(\mathcal{\mathcal{\tilde S}}_{l})\right\}.\]
\item
If we have classes $D'^{m}_{n (\cup_{l}At(\mathcal{\mathcal{\tilde S}}_{l}))}(M)$ and $D'^{m}_{n (\cup_{l}At(\mathcal{\mathcal{\tilde S}}_{l})o)}(M)$ where $At(\mathcal{\mathcal{\tilde S}}_{l})=\mathcal{\mathcal{\mathcal{\mathcal{S}}}}_{l}\subset\mathcal{\mathcal{\tilde S}}_{l}$, we use the simple notation ~$D'^{m}_{n (\cup_{l} \mathcal{\mathcal{\mathcal{\mathcal{S}}}}_{l})}(M)$, $D'^{m}_{n (\cup_{l} \mathcal{\mathcal{\mathcal{\mathcal{S}}}}_{l}o)}(M)$.\footnote{ We have to realize that the subatlas $\mathcal{\mathcal{\mathcal{\mathcal{S}}}}_{n}$ specifies
completely the atlas $\mathcal{\mathcal{\tilde S}}_{n}$, since taking forms smooth in $\mathcal{\mathcal{\mathcal{\mathcal{S}}}}_{n}$ determines automatically the whole
set of charts in which they are still smooth. This fact contributes to the simplicity of this notation.}
\end{itemize}

\subsubsection{Definition of $D'^{m}_{n A}(M)$, hence \emph{generalized tensor
fields}}

\theoremstyle{definition}
\newtheorem{defn8}[defn6]{Definition}
\begin{defn8}
 Now define~$D'^{m}_{n (\mathcal{\mathcal{\mathcal{\mathcal{S}}}}) A}(M)$ to be the algebra constructed from the objects ~$D'^{m}_{n (\mathcal{\mathcal{\mathcal{\mathcal{S}}}})}(M)$ by the tensor product,
exactly in analogy to the case of scalars (this reduces for scalars to the
product already defined). The
object, being a result of the tensor product, is again a mapping ~$V\to\mathbb{R}^{4^{m+n}}$, defined in every chart by componentwise multiplication.
Now denote by $D'^{m}_{n A}(M)$ a set given as $\cup_{\mathcal{\mathcal{\mathcal{\mathcal{S}}}}} D'^{m}_{n (\mathcal{\mathcal{\mathcal{\mathcal{S}}}}) A}(M)$, meaning a union of all possible $D'^{m}_{n (\mathcal{\mathcal{\mathcal{\mathcal{S}}}}) A}(M)$. Call the objects belonging to this set \emph{the generalized tensor fields (GTF)}.
\end{defn8}

\paragraph{Notation.}
Furthermore let us use the same procedure as in the previous
definition, just instead of constructing the algebras from
the classes $D'^{m}_{n (\mathcal{\mathcal{\mathcal{\mathcal{S}}}})}(M)$, we now construct them only from the classes ~~$D'^{m}_{n
(\cup_{l}At(\mathcal{\mathcal{\tilde S}}_{l}))}(M)\cap D'^{m}_{n (\mathcal{\mathcal{\mathcal{\mathcal{S}}}})}(M)$,
~~(again by tensor product). For the union of such algebras we use the notation
~$D'^{m}_{n (\cup_{l}At(\mathcal{\mathcal{\tilde S}}_{l})) A}(M)$.

\subsubsection{Definition of $\Gamma-$objects, their classes and algebras}

\paragraph{Notation.}
Now let us define the generalized space of objects $\Gamma^{l}(M)$. (For example the Christoffel symbol would fall into this class.) These objects are defined exactly in the same way as $D'^{m}_{n}(M)$~ ($m+n=l$)~ objects, we just do not require that they transform between charts in the tensorial way, (second point in the definition of generalized tensor fields).

Note the following:
\begin{itemize}
\item
The definition of $\Gamma^{l}(M)$ includes also the case
$m=0$. Now we see, that the scalars can be
taken as subclass of $\Gamma^{0}(M)$, given by objects that are constants with respect to the chart argument.
\item
Note also that for a general $\Gamma^{m}(M)$ object
there is no meaningful differentiation between ``upper'' and
``lower'' indices, but we will still use formally the $T^{\mu...}_{\nu...}$ notation (for all cases).
\end{itemize}

\paragraph{Notation.}
In the same way, (by just not putting requirements on the transformation properties), we can generalize the classes
\begin{itemize}
\item
 $D'^{m}_{n (\cup_{l}At(\mathcal{\mathcal{\tilde S}}_{l}))}(M)$ ~to~ $\Gamma^{m+n}_{(\cup_{l}At(\mathcal{\mathcal{\tilde S}}_{l}))}(M)$,
\item
~$D'^{m}_{n (\cup_{l}At(\mathcal{\mathcal{\tilde S}}_{l})o)}(M)$ ~to~ $\Gamma^{m+n}_{(\cup_{l}At(\mathcal{\mathcal{\tilde S}}_{l})o)}(M)$,
\item
 ~$D'^{m}_{n (\cup_{l}\mathcal{\mathcal{\mathcal{\mathcal{S}}}}_{l})}(M)$ ~to~ $\Gamma^{m+n}_{(\cup_{l}\mathcal{\mathcal{\mathcal{\mathcal{S}}}}_{l})}(M)$,
\item
 ~$D'^{m}_{n (\cup_{l} \mathcal{\mathcal{\mathcal{\mathcal{S}}}}_{l}o)}(M)$ ~to~  $\Gamma^{m+n}_{(\cup_{l}\mathcal{\mathcal{\mathcal{\mathcal{S}}}}_{l}o)}(M)$,
\item
~$D'^{m}_{n E}(M)$ ~to~ $\Gamma^{m+n}_{E}(M)$,
\item
~$D'^{m}_{n (\cup_{l}At(\mathcal{\mathcal{\tilde S}}_{l})) A}(M)$ ~to~ $\Gamma^{m+n}_{(\cup_{l}At(\mathcal{\mathcal{\tilde S}}_{l})) A}(M)$, ~and
\item
$D'^{m}_{n A}(M)$~ to ~$\Gamma^{m+n}_{A}(M)$.
\end{itemize}
\medskip

It is obvious that all the latter classes contain all the former classes as their subclasses, (this is a result of what we called a ``generalization'').

Note that when we fix $\Gamma^{m}_{E}(M)$ objects in arbitrary chart from $\mathcal{A}$, they must be expressed by integrals from multi-index matrix of functions integrable on every compact set. In the case of the $D'^{m}_{n E}(M)$ subclass it can be required in only one chart, since the transformation properties together with boundedness of Jacobians and inverse Jacobians, provide that it must hold in any other chart from $\mathcal{A}$.
The specific subclass of $\Gamma^{m}_{E}(M)$
is ~$\Gamma^{m}_{S}(M)$, which is a subclass of distributions given in any chart from $\mathcal{A}$, (being an argument
of the given $\Gamma-$ object), by integrals from multi-index matrix of smooth
functions (when we express the integrals in the same chart, as the one taken as
the argument). ~$\Gamma^{m}_{S (\cup_{n}\mathcal{\mathcal{\mathcal{\mathcal{S}}}}_{n} o)}(M)$ stands again for ~$\Gamma^{m}_{S}(M)$~ objects with domain limited to

\[\cup_{\Omega_{Ch}}\cup_{n}\left\{ (\omega,Ch(\Omega_{Ch})):~\omega\in C^{P}_{S (\mathcal{\mathcal{\tilde S}}_{n})}(\Omega_{Ch}),~
Ch(\Omega_{Ch})\in\mathcal{\mathcal{\mathcal{\mathcal{\mathcal{S}}}}}_{n}\right\},\]
 where $\mathcal{\mathcal{\tilde S}}_{n}$ is given by the condition $\mathcal{\mathcal{\mathcal{\mathcal{S}}}}_{n}\subset\mathcal{\mathcal{\tilde S}}_{n}$.

\paragraph{Notation.}
Take some arbitrary elements
~$T^{\mu...}_{\nu...}\in\Gamma^{m}_{E}(M)$, ~$\omega\in
C^{P}(\Omega_{Ch})$,~ and ~$Ch_{k}(\Omega_{Ch})\in\mathcal{A}$.
~The $T^{\mu...}_{\nu...}(Ch_{k},\omega)$ can be always expressed as
$\int_{\Omega_{Ch}}T^{\mu...}_{\nu...}(Ch_{k})\cdot\omega$. ~Here
$T^{\mu...}_{\nu...}(Ch_{k})$ appearing under the integral denotes
some multi-index matrix of functions continuous to a maximal possible degree
on $\Omega_{Ch}$. For $T^{\mu...}_{\nu...}\in D'^{m}_{n E}(M)$ the
$T^{\mu...}_{\nu...}(Ch_{k})$ multi-index matrix components can be obtained from a tensor
field by:
\begin{itemize}
\item
expressing the tensor field components in $Ch_{k}(\Omega_{Ch})$ on some subset of $\mathbb{R}^{4}$,
\item
mapping the tensor field components to $\Omega_{Ch}$ by the inverse of $Ch_{k}(\Omega_{Ch})$.
\end{itemize}
Furthermore $\omega'(Ch_{k})$ will denote the 4-form scalar density in the
chart $Ch_{k}(\Omega_{Ch})$.

\subsubsection{Topology on $\Gamma^{m}_{(\cup_{n}\mathcal{\mathcal{\mathcal{\mathcal{S}}}}_{n}o)}(M)$}

If we take the class of ~$\Gamma^{m}_{(\cup_{n}
\mathcal{\mathcal{\mathcal{\mathcal{S}}}}_{n} o)}(M)$, and we impose
on this class the weak (point or $\sigma-$) topology, then the
subclass of ~$\Gamma^{m}_{(\cup_{n}
\mathcal{\mathcal{\mathcal{\mathcal{S}}}}_{n} o)}(M)$~ defined as
~$\Gamma^{m}_{S (\cup_{n}
\mathcal{\mathcal{\mathcal{\mathcal{S}}}}_{n} o)}(M)$ is dense in
~$\Gamma^{m}_{(\cup_{n}
\mathcal{\mathcal{\mathcal{\mathcal{S}}}}_{n} o)}(M)$. The same
holds for ~$D'^{m}_{n (\cup_{l}
\mathcal{\mathcal{\mathcal{\mathcal{S}}}}_{l} o)}(M)$ and $D'^{m}_{n
S (\cup_{l}\mathcal{\mathcal{\mathcal{\mathcal{S}}}}_{l} o)}(M)$.

\subsubsection{Definition of contraction}

\theoremstyle{definition}
\newtheorem{defn85}[defn6]{Definition}
\begin{defn85}
We define the contraction of a ~$\Gamma^{m}_{A}(M)$ object in the expected way: It is a map that transforms the object ~$T^{...\mu...}_{...\nu...}\in\Gamma^{m}_{A}(M)$~ to~ the object $T^{...\mu...}_{...\mu...}\in\Gamma^{m-2}_{A}(M)$.
\end{defn85}

Now contraction is a mapping ~$D'^{m}_{n}(M)\to
D'^{m-1}_{n-1}(M)$ and ~$\Gamma^{m}(M)\to\Gamma^{m-2}(M)$, but it is
not in general the mapping~ $D'^{m}_{n A}(M)\to D'^{m-1}_{n-1
A}(M)$, only~  $D'^{m}_{n A}(M)\to\Gamma^{m+n-2}_{A}(M)$.

\subsubsection{Interpretation of physical quantities}

The interpretation of
physical observables as ``amounts'' of quantities on the open sets
is dependent on our notion of volume. So how shall we get the notion
of volume in the context of our language? First, by volume we mean a
volume of an open set. But we will consider only open sets belonging
to some $\Omega_{Ch}$. So take some ~$\Omega_{Ch}$~ and some
arbitrary ~$\Omega'\subset\Omega_{Ch}$.~ Let us now assume that we
have a metric tensor from ~$D'^{m}_{n E}(M)$. This induces a
(volume) 4-form. Multiply this 4-form by a noncontinuous function
$\chi_{\Omega'}$ defined to be 1 inside ~$\Omega'$~ and everywhere
else 0. Call it $\omega_{\Omega'}$. Then by volume of an open set
$\Omega'$ we understand: ~$\int\omega_{\Omega'}$. Also
$\omega_{\Omega'}$ is object from ~$C^{P}(M)$ (particularly from
$C^{P}(\Omega_{Ch})$). The ``amounts'' of physical quantities on
$\Omega'$ we obtain, when the ~$D'^{m}_{n A}(M)$ objects act on
$\omega_{\Omega'}$.

\subsection{The relation of equivalence ($\approx$)}

This section now provides us with the fundamental concept of the theory,
the concept of equivalence of generalized tensor fields. Most of the first
part is devoted to fundamental definitions, the beginning of the second part
deals with the basic, important theorems, which just
generalize some of the basic Colombeau theory results to the
tensor product of generalized tensor fields. It adds several
important conjectures as well. The first part ends with the subsection
``some additional definitions'' and the second part with the subsection ``some additional theory''. They both deal with much
less central theoretical results, but they serve very well to put
light on what equivalence of generalized tensor fields means ``physically''.

\subsubsection{The necessary concepts to define the equivalence relation}

\paragraph{Notation.}
Take some subatlas of our atlas, this will be a
maximal subatlas of charts, which are maps to the whole of
$\mathbb{R}^{4}$. Such maps exist on each set $\Omega_{Ch}$ and they
will be denoted as $Ch'(\Omega_{Ch})$. We say that a chart $Ch'(\Omega_{Ch})$ is centered at the point $q\in\Omega_{Ch}$, if this point is mapped by this chart to 0 (in
$\mathbb{R}^{4}$). We will use the
notation $Ch'(q,\Omega_{Ch})$.

\paragraph{Notation.}
Take some $\Omega_{Ch}$,~$q\in\Omega_{Ch}$~ and
~$Ch'(q,\Omega_{Ch})\in\mathcal{\tilde S}$. The set of 4-forms
~$\omega_{\epsilon}\in A^{n}(\mathcal{\tilde
S},Ch'(q,\Omega_{Ch}))$ is defined in such way that
~$\omega_{\epsilon}\in C^{P}_{S (\mathcal{\mathcal{\tilde
S}})}(\Omega_{Ch})$ belongs to this class if:
\begin{itemize}
\item[a)] in the given $Ch'(q,\Omega'_{Ch})$, $\forall\epsilon$~ it holds that:

~~~~~$\int_{Ch'(\Omega'_{Ch})}(\prod_{i}x^{k_{i}}_{i})~\omega'_{\epsilon}(x)~d^{4}x=\delta_{k
0},~~ \sum_{i} k_{i}=k$,~~ $k\leq n$, ~~$n\in\mathbb{N}$,

\item[b)] the dependence on $\epsilon$ is in $Ch'(q,\Omega'_{Ch})$ given as $\epsilon^{-4}\omega'(\frac{x}{\epsilon})$.
\end{itemize}

\paragraph{Notation.}
Take an arbitrary ~$q$, ~~$\Omega_{Ch}$ ~($q\in\Omega_{Ch}$),
~~$Ch'(q,\Omega_{Ch})\in\mathcal{\mathcal{\tilde S}}$ ~and some natural number $n$. For
any ~$\omega_{\epsilon}\in A^{n}(Ch'(q,\Omega_{Ch}),\mathcal{\mathcal{\tilde S}})$ ~we can,
relatively to ~$Ch'(\Omega_{Ch})$, ~define a continuous set of maps
(depending on the parameter $y$)
\begin{equation}
A^{n}(Ch'(q,\Omega_{Ch}),\mathcal{\mathcal{\tilde S}})\to
C^{P}_{S (\mathcal{\mathcal{\tilde S}})}(\Omega_{Ch}),\label{basic}
\end{equation}
such that they are,  on ~$\Omega_{Ch}$ ~and in ~$Ch'(\Omega_{Ch})$, given as
~~$\omega'(\frac{x}{\epsilon}){\epsilon^{-4}}\to\omega'(\frac{y-x}{\epsilon})\epsilon^{-4}$.
~(To remind the reader $\omega'$ is the density expressing
 $\omega$ in this chart.) This gives us (depending on the parameter ~$y\in\mathbb{R}^{4}$) various $C^{P}_{S (\mathcal{\mathcal{\tilde S}})}(\Omega_{Ch})$ objects, such that they are in the fixed $Ch'(\Omega_{Ch})$ expressed by ~$\omega'(\frac{x-y}{\epsilon})~\epsilon^{-4}dx^{1}\bigwedge ...\bigwedge dx^{4}$. Denote these 4-form fields by $\tilde\omega_{\epsilon}(y)$.

\paragraph{Notation.}
Now, take any ~$T^{\mu...}_{\nu...}\in\Gamma^{m}_{ (At(\mathcal{\mathcal{\tilde S}}))A}(M)$. By applying it in an arbitrary fixed chart  ~~$Ch_{k}(\Omega_{Ch})\in At(\mathcal{\mathcal{\tilde S}},\Omega_{Ch})$~~
 on the 4-form field ~~$\tilde\omega_{\epsilon}(y)$,~~ obtained from ~~$\omega_{\epsilon}\in$ $A^{n}(Ch'(q,\Omega_{Ch}),\mathcal{\tilde
 S})$~ through the map (\ref{basic}), we get a function ~$\mathbb{R}^{4}\to\mathbb{R}^{4^{m+n}}$. As a consequence, the resulting function depends on the following objects:
~$T^{\mu...}_{\nu...}\left(\in\Gamma^{m}_{A}(M)\right),~\omega_{\epsilon}\left(\in
A^{n}(q,\Omega_{Ch},Ch'(\Omega_{Ch}))\right)$ ~and
$~Ch_{k}(\Omega_{Ch})$.~ We denote it by:
~~~~~$F'^{\mu...}_{\nu...}\big(T^{\mu...}_{\nu...}, \mathcal{\mathcal{\tilde S}},
~\Omega_{Ch}, Ch'(q,\Omega_{Ch}),n,\tilde\omega_{\epsilon}(y),
Ch_{k}(\Omega_{Ch})\big)$.

\subsubsection{Definition of the equivalence relation}

\theoremstyle{definition}
\newtheorem{defn11}{Definition}[section]
\begin{defn11}
 ~~$B^{\mu...}_{\nu...},T^{\mu...}_{\nu...}\in \Gamma^{m}_{A}(M)$
~are called equivalent ~($B^{\mu...}_{\nu...}\approx
T^{\mu...}_{\nu...}$), if:
\begin{itemize}
\item
 they belong to the same classes
~$\Gamma^{m}_{(At(\mathcal{\mathcal{\tilde S}}))A}(M)$,
\item
~$\forall\Omega_{Ch},~~~\forall q~~(q\in\Omega_{Ch}),~~~\forall
Ch'(q,\Omega_{Ch})\in\mathcal{\mathcal{\tilde S}}$ ~~~~(~such ~that ~~$B^{\mu...}_{\nu...}$,
~$T^{\mu...}_{\nu...}~\in\\ ~~\Gamma^{m}_{ (At(\mathcal{\mathcal{\tilde S}}))A}(M)$ ),~
~~$\forall Ch(\Omega_{Ch})~\in ~At(\mathcal{\mathcal{\tilde S}},\Omega_{Ch})$~~ $\exists n$,~
~~~such ~~that ~~~ $\forall \omega_{\epsilon}~\in \\
A^{n}(Ch'(q,\Omega_{Ch}),\mathcal{\mathcal{\tilde S}})$ ~~and for any compactly supported,
smooth function\\ $\mathbb{R}^{4}\to\mathbb{R}$, ~~$\phi$,~~ it holds:
\begin{eqnarray}
\lim_{\epsilon\to 0}
~\int_{\mathbb{R}^{4}}\bigg\{F'^{\mu...}_{\nu...}\big(B^{\mu...}_{\nu...}, q,
\Omega'_{Ch}, Ch'(\Omega'_{Ch}),n,\tilde\omega_{\epsilon}(y),
Ch(\Omega_{Ch})\big) \nonumber~~~~~~~~~~~~~~~~~~~~~~~~~~~~~~~\\
- ~F'^{\mu...}_{\nu...}\big(T^{\mu...}_{\nu...}, q, \Omega'_{Ch},
Ch'(\Omega'_{Ch}),n,\tilde\omega_{\epsilon}(y), Ch(\Omega_{Ch})\big)\bigg\} \cdot~\phi(y)
~d^{4}y~=~0.~~~~~~~
\end{eqnarray}
\end{itemize}
\end{defn11}

Note that for
$B^{\mu...}_{\nu...},C^{\mu...}_{\nu...},D^{\mu...}_{\nu...},T^{\mu...}_{\nu...}$
having the same domains and being from the same
$\Gamma^{n}_{(At(\mathcal{\tilde S}))}(M)$ classes, it trivially
follows that: $T^{\mu...}_{\nu...}\approx B^{\mu...}_{\nu...}$,~
$C^{\mu...}_{\nu...}\approx D^{\mu...}_{\nu...}$~ implies
~$\lambda_{1}T^{\mu...}_{\nu...}+\lambda_{2}C^{\mu...}_{\nu...}\approx\lambda_{1}B^{\mu...}_{\nu...}+\lambda_{2}D^{\mu...}_{\nu...}$~
for ~$\lambda_{1}, \lambda_{2}\in\mathbb{R}$.

\paragraph{Notation.}
Now since we have defined an equivalence relation, it divides the objects ~$\Gamma^{m}_{A}(M)$ naturally into equivalence classes. The set of such equivalence classes will be denoted as ~$\tilde\Gamma^{m}_{A}(M)$.
Later we may also use sets of more limited classes of
equivalence ~$\tilde D'^{m}_{n A}(M)$, ~$\tilde D'^{m}_{n E A}(M)$
(etc.), which contains equivalence classes (only) of the objects
belonging to ~$D'^{m}_{n A}(M)$, ~$D'^{m}_{n E A}(M)$
(etc.).

\paragraph{Notation.}
In some of the following theorems, (also for example in the definition of the covariant derivative), we will use some convenient notation:
Take some object $B^{\mu...}_{\nu...}\in\Gamma^{m}_{E}(M)$. The
expression ~$T^{\mu...}_{\nu...}(B^{\alpha...}_{\beta...}\omega)$ will be understood in the following way: Take ~$Ch_{k}(\Omega_{Ch}) \times \omega$ ~~($\omega\in C^{P}(\Omega_{Ch})$) from the domain of $T^{\mu...}_{\nu...}$. ~Then ~$B^{\alpha...}_{\beta...}(Ch_{k})\cdot\omega$~ is a multi-index matrix of $C^{P}(\Omega_{Ch})$ objects. This means that outside $\Omega_{Ch}$ set they are defined to be trivially 0. We substitute this multi-index matrix of $C^{P}(\Omega_{Ch})$ objects
to ~$T^{\mu...}_{\nu...}$, with the chart ~$Ch_{k}(\Omega_{Ch})$ taken as the argument.

\subsubsection{Relation to Colombeau equivalence}

A careful reader now understands the relation between our concept of equivalence and the Colombeau equivalence relation. It is simple: The previous definition just translates the Colombeau equivalence relation (see \cite{Amulti}) into our language and the equivalence classes will naturally preserve all the features of the Colombeau equivalence classes (this will be proven in the following theorems).

\subsubsection{Some additional definitions (concepts of associated field and
$\Lambda$ class)}

We define the concept (of association) to bring
some insight to what our concepts mean in the most simple (but most
important and useful) cases. It enables us to see better
the relation between the calculus we defined (concerning
equivalence) and the classical tensor calculus. It brings us also
better understanding of what equivalence means in
terms of physics (at least in the simple cases). It just means that the quantities might differ on the
large scales, but take the same small scale limit (for the small
scales they approach each other).

\theoremstyle{definition}
\newtheorem{defn12}[defn11]{Definition}
\begin{defn12}\label{DefAssoc}
 Take $T^{\mu...}_{\nu...}\in\Gamma^{m}_{A}(M)$. Assume that:
\begin{itemize}
 \item [a)]
~~$\forall \mathcal{\mathcal{\mathcal{\mathcal{S}}}}$,~~~ such~~ that ~~~$T^{\mu...}_{\nu...}~\in~\Gamma^{m}_{( \mathcal{\mathcal{\mathcal{\mathcal{S}}}})A}(M)$, ~~~$\forall\Omega_{Ch}$~~~ and ~~~$\forall Ch_{k}(\Omega_{Ch})~\in ~\mathcal{\mathcal{\mathcal{\mathcal{S}}}}$ \\  $\exists~\Omega_{Ch}~\setminus ~\Omega'(Ch_{k})$,~~~(the set ~$\Omega'(Ch_{k})$~ being ~0~ in
any Lebesgue measure),
~~such ~that ~~$\forall q~\in~\Omega_{Ch}~\setminus ~\Omega'(Ch_{k})$, ~~~$\forall Ch'(q,\Omega_{Ch})~\in ~\mathcal{\mathcal{\mathcal{\mathcal{S}}}}~\subset~\mathcal{\mathcal{\tilde S}}$  ~~~~$\exists n$,\\~ such ~that  ~$\forall \omega_{\epsilon}\in A^{n}(Ch'(q,\Omega_{Ch}),\mathcal{\mathcal{\tilde S}})$
\begin{equation}
\exists~\lim_{\epsilon\to 0}
T^{\mu...}_{\nu...}(Ch_{k},\omega_{\epsilon}).\label{limit}
\end{equation}
\item [b)]
The limit (\ref{limit}) is
~~$\forall Ch'(q,\Omega_{Ch})~\in ~\mathcal{\mathcal{\mathcal{\mathcal{S}}}}~\subset~\mathcal{\mathcal{\tilde S}}$,~ ~$\forall\omega_{\epsilon}~\in ~A^{n}(Ch'(q,\Omega_{Ch}),\mathcal{\mathcal{\tilde S}})$ ~the same.
\end{itemize}
 If both $a)$ and $b)$ hold, then the object defined by the limit (\ref{limit})
 is a mapping:
\begin{equation}
Ch(\Omega_{Ch})(\in\mathcal{\mathcal{\mathcal{\mathcal{\mathcal{S}}}}}))\times\Omega_{Ch}\setminus \Omega'(Ch)\to\mathbb{R}^{4^{m+n}}.
\end{equation}
 We call this map the field associated to ~$T^{\mu...}_{\nu...}\in\Gamma^{m}_{A}(M)$,
 and we use the expression
 ~$A_{s}(T^{\mu...}_{\nu...})$. (It necessarily fulfills the same consistency conditions for ~$\Omega^{1}_{Ch}\subset\Omega^{2}_{Ch}$ as the $\Gamma^{m}_{A}(M)$ objects.)

\end{defn12}

\theoremstyle{definition}
\newtheorem{deff12}[defn11]{Definition}
\begin{deff12}
Denote by ~$\Lambda\subset\Gamma^{m}_{E (\cup_{n}At(\mathcal{\mathcal{\tilde S}}_{n}) o)}(M)$ a class of objects, such that
each ~$T^{\mu...}_{\nu...}\in\Lambda$~~~ can be
~~~$\forall\Omega_{Ch}$,~~~ $\forall n$,~~~ $\forall\omega\in C^{P}_{S (\mathcal{\mathcal{\tilde S}}_{n})}(\Omega_{Ch})$, ~~~$\forall \mathcal{\mathcal{\mathcal{\mathcal{S}}}}\subset\mathcal{\mathcal{\tilde S}}_{n}~\cap ~At(\mathcal{\mathcal{\tilde S}}_{n})$,\\~~ $\forall Ch_{k}(\Omega_{Ch})\in\mathcal{\mathcal{\mathcal{\mathcal{\mathcal{S}}}}}$~ expressed as a map
\begin{equation}
(\omega,Ch_{k})\to\int_{\Omega_{Ch}}T^{\mu...}_{\nu...}(Ch_{k})\cdot\omega,
\end{equation}
where for $T^{\mu...}_{\nu...}(Ch_{k})$ holds the following: In each
chart from $\mathcal{\mathcal{\mathcal{\mathcal{\mathcal{S}}}}}$ for
every point ~$\mathbf{z_{0}}$, ~where $T^{\mu...}_{\nu...}(Ch_{k})$
is continuous ~$\exists~\delta>0,~\exists~ K^{\mu}_{\nu}>0$, such
that ~$\forall\epsilon$~ $(0\leq\epsilon\leq\delta))$ and for
arbitrary unit vector $\mathbf{n}$~(in the Euclidean metric on
$\mathbb{R}^{4}$)~~
\begin{equation}
T^{\mu...}_{\nu...}(Ch_{k},\mathbf{z_{0}})-K^{\mu...}_{\nu....}\epsilon\leq
T^{\mu...}_{\nu...}(Ch_{k},\mathbf{z_{0}}+\mathbf{n}\epsilon)\leq
T^{\mu...}_{\nu...}(Ch_{k},\mathbf{z_{0}})+K^{\mu...}_{\nu...}\epsilon.
\end{equation}
\label{Lambda}
\end{deff12}

\paragraph{Notation.}
Take from (\ref{Lambda}) arbitrary, fixed ~$T^{\mu...}_{\nu...}$,
~$\Omega_{Ch}$ ~and ~$Ch_{k}(\Omega_{Ch})$. ~By the notation ~$\tilde\Omega(Ch_{k})\subset\Omega_{Ch}$ we denote a set (having Lebesgue measure 0) on
which is  ~$T^{\mu...}_{\nu...}(Ch_{k})$ discontinuous.

\subsubsection{Reproduction of the basic results by the equivalence relation}

\newtheorem{uniqueness}{Theorem}[section]

\begin{uniqueness}
 Any class ~$\tilde\Gamma^{m}_{(\cup_{n}At(\mathcal{\mathcal{\tilde S}}_{n}) o)A}(M)$ contains maximally one linear element.
\end{uniqueness}

\begin{proof}
We need to prove that there do not exist such two elements of
~$\Gamma^{m}(M)$, which are equivalent. Take two elements $B$ and $T$
from the class $\Gamma^{m}_{(\cup_{n}At(\mathcal{\mathcal{\tilde S}}_{n})o)A}(M)$  (both with the given domains and continuity). Take
arbitrary ~$\Omega_{Ch}$, arbitrary ~$\mathcal{\mathcal{\tilde S}}$ from their domains,
and arbitrary ~$Ch'(\Omega_{Ch})\in\mathcal{\mathcal{\tilde S}}$. Map all the
~$C^{P}_{S (\mathcal{\mathcal{\tilde S}})}(\Omega_{Ch})$ objects to smooth, compact
supported functions on $\mathbb{R}^{4}$ through this fixed chart
mapping. Now both $B$ and $T$ give, in fixed but arbitrary
~$Ch(\Omega_{Ch})\in At(\mathcal{\mathcal{\tilde S}})$ linear, continuous maps on the compactly supported smooth
functions. (The only difference from Colombeau distributions is that
it is in general a map to $\mathbb{R}^{m}$, so the difference is only ``cosmetic''.)

Now after applying this construction, our
concept of equivalence reduces for every
~$Ch(\Omega_{Ch})\in At(\mathcal{\mathcal{\tilde S}})$ to Colombeau equivalence from
\cite{Amulti}. The same results must hold. One of the
results says that there are no two distributions being
equivalent. All the parameters are fixed but arbitrary and
all the 4-forms from domains of $B$ and $T$ can be mapped to the
~$\mathbb{R}^{4}$ functions for some proper fixing of ~$\Omega_{Ch}$
and $\mathcal{\mathcal{\tilde S}}$. Furthermore, the $C^{P}_{S}(M)$ 4-forms are arguments of $B$ and $T$ only in the
charts, in which $B$ and $T$ were compared as maps on the spaces of
~$\mathbb{R}^{4}$ functions. So this ``arbitrary
chart fixing'' covers all their domain. As a result $B$ and $T$ must be identical and that is what needed to be proven.
\end{proof}

\newtheorem{extension}[uniqueness]{Theorem}
\begin{extension}
 Any class of equivalence ~$\tilde\Gamma^{m}_{E A}(M)$ contains maximally one linear element.
\end{extension}

\begin{proof}
 First notice that the elements of ~$\Gamma^{m}_{E A}(M)$ are continuous and
 defined on every ~$C^{P}_{S (\mathcal{\mathcal{\tilde S}})}(M)$ in every
chart from $\mathcal{A}$, so they are required to be compared in any arbitrary chart from $\mathcal{A}$. By taking this into account, we can repeat the previous proof. There is one additional trivial fact one has to notice:
the ~$C^{P}_{S}(M)$ domain also uniquely determines how the ~$\Gamma^{m}_{E}(M)$ element acts
outside ~$C^{P}_{S}(M)$. So if $B,~T\in \Gamma^{m}_{E A}(M)$ give the same map on ~$C^{P}_{S}(M)$, they give the same map everywhere.
\end{proof}

\newtheorem{multiplication}[uniqueness]{Theorem}

\begin{multiplication}\label{multiplication}
The following statements hold:
\begin{itemize}
\item[a)] Take ~$T^{\mu...\alpha}_{\nu...\beta}\in\Gamma^{a}_{E A}(M)$ ~such that ~$\forall\Omega_{Ch}$,
~$\forall Ch_{k}(\Omega_{Ch})\in\mathcal{A}$ ~and ~$\forall\omega\in C^{P}(\Omega_{Ch})$ ~~$T^{\mu...\alpha}_{\nu...\beta}$~ is defined as a map
\begin{equation}
(Ch_{k},\omega)\to\int_{\Omega_{Ch}}
T^{\mu...}_{1~\nu...}(Ch_{k})~\omega~...~\int_{\Omega_{Ch}} T^{\alpha...}_{N \beta...}(Ch_{k})~\omega~.\label{Teq}
\end{equation}
Then the class of equivalence ~$\tilde\Gamma^{a}_{E A}(M)$,
to which ~$T^{\mu...}_{\nu...}$~ belongs,
contains a linear element defined (on arbitrary ~$\Omega_{Ch}$) as the map: ~~$\forall\omega~\in ~C^{P}(\Omega_{Ch})$, ~~$\forall Ch_{k}(\Omega_{Ch})~\in ~\mathcal{A}$,
\begin{eqnarray}
(Ch_{k},\omega)~\to~\int_{\Omega_{Ch}}
T^{\mu...}_{1~\nu...}(Ch_{k})...T^{\alpha...}_{N \beta...}(Ch_{k})~\omega~,
\end{eqnarray}
if and only if ~~$\forall\Omega_{Ch}$,~~ $\forall Ch_{k}(\Omega_{Ch})~\in ~\mathcal{A}$ ~~~ $\exists ~Ch_{l}(\Omega_{Ch})~\in ~\mathcal{A}$, ~~such that
\begin{eqnarray}
\int_{Ch_{l}(\Omega_{Ch})_{|\Omega'}} T^{\mu...}_{1~\nu...}(Ch_{k})...T^{...\alpha}_{N...\beta}(Ch_{k}) ~d^{4}x~~~~~~~~~~~~~~~~~~~~~~~~~
\end{eqnarray}
converges on every
compact set $\Omega'\subset\Omega_{Ch}$.

The same statement holds,
if we take instead of ~$\Gamma^{m}_{E A}(M)$ its subclass ~$D'^{a}_{b
E A}(M)$ and instead of the equivalence class ~$\tilde\Gamma^{m}_{E
A}(M)$, the equivalence class ~$\tilde D'^{a}_{b E A}(M)$.

The same statement also holds if we take instead of ~$\Gamma^{m}_{E
A}(M)$ and ~$D'^{m}_{n E A}(M)$ classes, the classes ~$\Gamma^{m}_{E
A (\cup_{l}At(\mathcal{\mathcal{\tilde S}}_{l}) o)}(M)$ and
~$D'^{m}_{n E A (\cup_{l}At(\mathcal{\mathcal{\tilde
S}}_{l})o)}(M)$, ~(with the exception that the given convergence
property shall be considered only for charts from ~
$\cup_{l}At(\mathcal{\tilde S}_{l})$ ).
\item[b)] For any distribution ~~$A^{\alpha...}_{\beta...}~\in ~\Gamma^{a}_{ S
(\cup_{n} \mathcal{\mathcal{\mathcal{\mathcal{S}}}}_{n}~ o)}(M)$,~~ and an element
~~$T^{\mu...}_{\nu...}\in\\
\Gamma^{m}_{(\cup_{n} \mathcal{\mathcal{\mathcal{\mathcal{S}}}}_{n})}(M)$, ~~we have
that ~~$A^{\alpha...}_{\beta...}T^{\mu...}_{\nu...}$~~ is equivalent
to an element of ~~$\Gamma^{m+a}_{(\cup_{n} \mathcal{\mathcal{\mathcal{\mathcal{S}}}}_{n} ~o)}(M)$, ~~(and~
for ~subclasses ~~$D'^{a}_{b
(\cup_{n}\mathcal{\mathcal{\mathcal{\mathcal{S}}}}_{n})}(M)~\subset~\Gamma^{a+b}_{(\cup_{n}\mathcal{\mathcal{\mathcal{\mathcal{S}}}}_{n})}(M)$ ~~and~~\\
$D'^{k}_{l S( \cup_{n}\mathcal{\mathcal{\mathcal{\mathcal{S}}}}_{n}
o)}(M)~\in~\Gamma^{k+l}_{S (\cup_{n}
\mathcal{\mathcal{\mathcal{\mathcal{S}}}}_{n} ~o)}(M)$ ~~it ~is
~equivalent ~to ~an ~element ~of ~~$D'^{k+a}_{l+b
(\mathcal{\mathcal{\mathcal{\mathcal{S}}}} ~o)}(M)$).~ The element
is on its domain given as the mapping
\begin{equation}
(\omega,Ch_{k})\to T^{\mu...}_{\nu...}(A^{\alpha...}_{\beta...}\omega).
\end{equation}
\item[c)] For ~any ~tensor ~distribution ~~$A^{\alpha...}_{\beta...}~\in
~\Gamma^{a}_{S}(M)$  ~~~and ~an ~element
~~$T^{\mu...}_{\nu...}~\in\\
\Gamma^{m}_{(\cup_{n} \mathcal{\mathcal{\mathcal{\mathcal{S}}}}_{n} ~o)}(M)$,
~we ~have ~that ~~$A^{\alpha...}_{\beta...}T^{\mu...}_{\nu...}$~~ is
equivalent ~to ~an ~element ~of ~~\\$\Gamma^{m+a}_{(\cup_{n} \mathcal{\mathcal{\mathcal{\mathcal{S}}}}_{n}
~o)}(M)$. ~The element is on its domain given as mapping
\begin{equation}
(\omega,Ch_{k})\to T^{\mu...}_{\nu...}(A^{\alpha...}_{\beta...}\omega).
\end{equation}
\end{itemize}
\end{multiplication}

\begin{proof}
\begin{itemize}
\item[a)]Use exactly the same construction as in the previous proof. For arbitrary $\Omega_{Ch}$ and arbitrary $Ch_{k}(\Omega_{Ch})\in\mathcal{A}$,
we see that ~$T^{\mu...}_{\nu...}$ is for every $\omega\in C^{P}_{S (\mathcal{\mathcal{\tilde S}})}(\Omega_{Ch})$ given by (\ref{Teq}) (it is
continuous in arbitrary chart on every ~$C^{P}_{S (\mathcal{\mathcal{\tilde S}})}(\Omega_{Ch})$).
 We can express the map (\ref{Teq}) in some chart
$Ch_{l}(\Omega_{Ch})$ as ~
\begin{eqnarray}
(Ch_{k},\omega')~~~~~~~~~~~~~~~~~~~~~~~~~~~~~~~~~~~~~~~~~~~~~~~~~~~~~~~~~~~~~~~~~~~~~~~~~~~~~~~~~~~~~~~~~~~~~~~~\nonumber\\~
\to~\int_{Ch_{l}(\Omega_{Ch})}
T^{\mu...}_{1~\nu...}(Ch_{k})~\omega'~d^{4}x...\int_{Ch_{l}(\Omega_{Ch})}T^{\alpha...}_{N \beta...}(Ch_{k})~\omega'~d^{4}x.~~~~\label{Teq2}
\end{eqnarray}

Then it is a result of Colombeau
theory that if
\begin{equation}
(Ch_{k},\omega')~\to~\int_{Ch_{l}(\Omega_{Ch})}
T^{\mu...}_{1~\nu...}(Ch_{k})...T^{...\alpha}_{N...\beta}(Ch_{k})~\omega'~d^{4}x\label{Teq3}
\end{equation}
is defined as
a linear mapping on compactly supported, smooth ~$\mathbb{R}^{4}$ ~functions ~$\omega'$ ~(in our case
they are related by $Ch_{l}(\Omega_{Ch})$ to given
~$C^{P}_{S}(M)$ objects), it is equivalent to (\ref{Teq2}). Now
everything was fixed, but arbitrary, so the result is
proven. From this proof we also see that the simple
transformation properties of the ~$D'^{m}_{n}(M)$
objects\footnote{We include also the scalar objects here.} are fulfilled by the map (\ref{Teq3}) if the objects multiplied are from ~$D'^{m}_{n
A}(M)$. So the second result can be proven immediately. The last two
results concerning the classes with limited domains trivially follow from the previous proof.
\item[b)] is proven completely in the same way, we just have to understand
that because of the ``limited'' domain of the ~$D'^{a}_{b
S(\cup_{n}\mathcal{\mathcal{\mathcal{\mathcal{S}}}}_{n} o)}(M)$ objects, we can effectively use the concept
of smoothness in this case.
\item[c)] is just the same as b), the only difference is that the domain
of the product is limited because of the ``second'' term in the
product.
\end{itemize}
\end{proof}

Note that this means that tensor product gives, on appropriate
subclasses of ~$D'^{m}_{n E A}(M)$, the mapping ~$\tilde D'^{a}_{b E A}(M)\times\tilde D'^{m}_{n E A}\to \tilde D'^{a+m}_{b+n E A}(M)$. ~It also means that this procedure gives, on
appropriate subclasses of $\Gamma^{m}_{E A}(M)$, the mapping ~$\tilde \Gamma^{a}_{E A}(M)\times\tilde
\Gamma^{m}_{E A}\to \tilde \Gamma^{a+m}_{E A}(M)$.
The disappointing fact is that this cannot be extended to ~$D'^{m}_{n
A}(M)$.

\newtheorem{Comm2}[uniqueness]{Theorem}
\begin{Comm2}
Take ~$T^{\mu...}_{\nu...}\in\Gamma^{m}_{(\cup_{n}
\mathcal{\mathcal{\mathcal{\mathcal{S}}}}_{n}~ o)A}(M)$,~
$B^{\mu...}_{\nu...}\in\Gamma^{m}_{(\cup_{n}
\mathcal{\mathcal{\mathcal{\mathcal{S}}}}_{n}~ o)}(M)$ and
~$L^{\alpha...}_{\beta...}\in\Gamma^{n}_{S (\cup_{l}
\mathcal{\mathcal{\mathcal{\mathcal{S}}}}_{l}~ o)}(M)$. Then
~$T^{\mu...}_{\nu...}\approx B^{\mu...}_{\nu...}$ implies
\footnote{It is obvious that we can extend the definition domains
either of $T^{\mu...}_{\nu...}$ and $B^{\mu...}_{\nu...}$, or of
$L^{\mu...}_{\nu...}$.} ~$(L\otimes
T)^{\alpha...\mu...}_{\beta...\nu...}\approx (L\otimes
B)^{\alpha...\mu...}_{\beta...\nu...}$.
\end{Comm2}

\begin{proof}
Use the same method as previously. It trivially follows from the
results of Colombeau theory (especially from the theorem saying that
if a Colombeau algebra object is equivalent to a distribution, then
after multiplying each of them by a smooth distribution, they remain
equivalent).
\end{proof}

\newtheorem{Contraction}[uniqueness]{Theorem}

\begin{Contraction}
 Contraction (of $\mu$ and $\nu$ index) is always, for such objects
~$T^{...\mu...}_{...\nu...}\in D'^{m}_{n E A}(M)$ that they are
equivalent to some linear element, a map to some element of the
equivalence class from ~~~$\tilde\Gamma^{m+n-2}_{E A}(M)$.~~~ The
~equivalence ~class ~from \\ $\tilde\Gamma^{m+n-2}_{E A}(M)$ is
such, that it contains (exactly) one element from ~$D'^{m-1}_{n-1
E}(M)$ and this element is defined as the map:
~~$\forall\Omega_{Ch}$, ~~$\forall Ch_{k}(\Omega_{Ch})~\in
~\mathcal{A}$, ~~$\forall\omega~\in ~C^{P}(\Omega_{Ch})$,
\begin{eqnarray}
(\omega,Ch_{k})\to\int_{\Omega_{Ch}}
T^{...\alpha...}_{...\alpha...}(Ch_{k})~\omega.~~~~~~~~~~~~
\end{eqnarray}
\end{Contraction}

\begin{proof}
The proof trivially follows from the fact that
contraction commutes with the relation of equivalence (this trivially follows from our previous note about
addition and equivalence).
\end{proof}

\subsubsection{Some interesting conjectures}

\theoremstyle{definition}
\newtheorem{ADD1}{Conjecture}[section]
\begin{ADD1}
Tensor product gives these two maps:
\begin{itemize}
\item
~$\tilde D'^{a}_{b E A}(M)\times\tilde D'^{m}_{n E A}(M)\to
\tilde D'^{a+m}_{b+n E A}(M)$,
\item
~$\tilde \Gamma^{a}_{E A}(M)\times\tilde \Gamma^{b}_{E
A}(M)\to \tilde \Gamma^{a+b}_{E A}(M)$.
\end{itemize}
\end{ADD1}

\theoremstyle{definition}
\newtheorem{ADD2}[ADD1]{Conjecture}
\begin{ADD2}
Take some
~$B^{\mu...}_{\nu...}\in\Gamma^{a}_{(\cup_{n}At(\mathcal{\tilde
S}_{n}) ~o)}(M)$. Take an element
~$T^{\mu...}_{\nu...}\in\Gamma^{b}_{E}(M)$, such that
~$\forall\mathcal{\mathcal{\tilde
S}}\subseteq\cup_{n}\mathcal{\mathcal{\tilde S}}_{n}$,
~$\forall\Omega_{Ch}$, ~$\forall Ch_{k}(\Omega_{Ch})\in
At(\mathcal{\mathcal{\tilde S}})$ it holds that
 ~$\forall\omega\in C^{P}_{S (\mathcal{\mathcal{\tilde S}})}(M)$, ~the elements of the multi-index matrix ~$T^{\mu...}_{\nu...}(Ch_{k})\cdot\omega$,~~
 are still
 from the class ~$C^{P}_{S (\mathcal{\mathcal{\tilde S}})}(M)\subset\cup_{n}C^{P}_{S
(\mathcal{\mathcal{\tilde S}}_{n})}(M)$. Then it holds that
~$B^{\alpha...}_{\beta...}T^{\mu...}_{\nu...}$ is equivalent to an
element of ~$\Gamma^{a+b}_{(\cup_{n}At(\mathcal{\mathcal{\tilde
S}}_{n}) ~o)}(M)$. (For subclasses ~$D'^{m}_{n
(\cup_{l}At(\mathcal{\mathcal{\tilde S}}_{l}) ~o)}(M)$ and
~$D'^{a}_{b E}(M)$ it is equivalent to an element ~$D'^{m+a}_{n+b
(\cup_{l}At(\mathcal{\tilde S}_{l}) ~o)}(M)$.) The element is on its
domain given as mapping
\begin{equation}
(\omega, Ch_{k})\to
B^{\mu...}_{\nu...}(T^{\alpha...}_{\beta...}\omega).
\end{equation}
\end{ADD2}

\subsubsection{Some additional theory}

\newtheorem{association1}{Theorem}[section]
\begin{association1}\label{Association1}
Any arbitrary ~$T^{\mu...}_{\nu...}\in\Lambda$ (as defined by \ref{Lambda}) defines an
~$A_{s}(T^{\mu...}_{\nu...})$ object on $M$. Take any arbitrary ~$\mathcal{\mathcal{\tilde S}}$~ from the domain of $T^{\mu...}_{\nu...}$ ,~any arbitrary ~$\mathcal{\mathcal{\mathcal{\mathcal{S}}}}\subset At(\mathcal{\mathcal{\tilde S}})\cap\mathcal{\mathcal{\tilde S}}$,~any arbitrary $\Omega_{Ch}$ and any arbitrary $Ch_{k}(\Omega_{Ch})\in\mathcal{\mathcal{\mathcal{\mathcal{\mathcal{S}}}}}$. Then for $\Omega_{Ch}\setminus\tilde\Omega(Ch_{k})$ it holds that multi-index matrix of functions $T^{\mu...}_{\nu...}(Ch_{k})$ can be obtained from the tensor components of $A_{s}(T^{\mu...}_{\nu...})$ in $Ch_{k}(\Omega_{Ch})$ by the inverse mapping to $Ch_{k}(\Omega_{Ch})$.
\end{association1}

\begin{proof}
 For ~$\forall\Omega_{Ch}$,~ take fixed but arbitrary ~$ Ch_{k}(\Omega_{Ch})\in\mathcal{\mathcal{\mathcal{\mathcal{\mathcal{S}}}}}$ and take
~$T^{\mu...}_{\nu...}(Ch_{k})$. Then ~$\forall q\in\Omega_{Ch}$, ~$\forall Ch'(\Omega_{Ch},q)\in\mathcal{\mathcal{\mathcal{\mathcal{\mathcal{S}}}}}$,  and ~$\forall ~\omega_{\epsilon}\in A^{n}(\mathcal{\mathcal{\tilde S}}_{l}, Ch'(q,\Omega_{Ch}))$, ~we see that  ~$\omega'(Ch')$~ is a
delta-sequence. That means we just have to show, that on the set where $T^{\mu...}_{\nu...}(Ch_{k})$ is continuous in the \ref{Lambda} sense, the delta-sequencies
give the value of this multi-index matrix. So write the
integral:
\begin{equation*}
\int_{Ch'(\Omega_{Ch})}
T^{\mu...}_{\nu...}\left[(Ch_{k}) (\mathbf{x})\right]~\frac{1}{\epsilon^{4}}~~\omega'\left(\frac{\mathbf{x}}{\epsilon}\right)~d^{4}x.
\end{equation*}
By substitution ~$\mathbf{x}=\epsilon.\mathbf{z}$ we obtain:

\begin{equation*}
\int_{Ch'(\Omega_{Ch})} T^{\mu...}_{\nu...}\left[(Ch_{k})~(\epsilon.\mathbf{
z})\right]~\omega'(\mathbf{z})~d^{4}z.
\end{equation*}
 But from the properties of
~$T^{\mu...}_{\nu...}\left[(Ch_{k})(\mathbf{x})\right]$ it follows that

\begin{eqnarray*}
\int_{Ch'(\Omega_{Ch})}(T^{\mu...}_{\nu...}\left[(Ch_{k})~(\mathbf
z_{0})-K^{\mu...}_{\nu...}\epsilon)\right]~\omega'(\mathbf{z})~d^{4}z~~~~~~~~~~~~~\\
\leq\int_{Ch'(\Omega_{Ch})}T^{\mu...}_{\nu...}\left[(Ch_{k})(\mathbf
z_{0}+\mathbf{n}\epsilon)\right]~\omega'(\mathbf{z})~d^{4}z~~~~~~~~\\
\leq\int_{Ch'(\Omega_{Ch})}(T^{\mu...}_{\nu...}\left[(Ch_{k})(\mathbf
z_{0})\right]+K^{\mu...}_{\nu...}\epsilon)~\omega'(\mathbf{z})~d^{4}z,
\end{eqnarray*}
 for some $\epsilon$ small
enough.

 But we are taking the limit ~$\epsilon\to 0$ which,
considering the fact that ~$\omega(\mathbf{x})$ are normed to 1,
means that the integral must give
~$T^{\mu...}_{\nu...}\left[(Ch_{k})(\mathbf{z_{0}})\right]$. The set, where it is not
continuous in the sense of \ref{Lambda}, has Lebesgue measure 0. That means the
~$\Omega_{Ch}$ part, which is mapped to this set has Lebesgue measure
0. But then the values of the multi-index matrix in the given chart
at this arbitrary, but fixed point give us an associated field (and are independent on delta sequence obviously).
\end{proof}

\newtheorem{association}[association1]{Theorem}

\begin{association}
 The field associated to
a ~$T^{\mu...}_{\nu...}\in \Lambda\cap D'^{m}_{n E}(M)$, ~transforms for each
~$\Omega_{Ch}$, ~for every pair of charts from its
domain, ~$Ch_{1}(\Omega_{Ch}), ~Ch_{2}(\Omega_{Ch})$ ~~on some ~$M/(\tilde\Omega(Ch_{1})\cup \tilde\Omega(Ch_{2}))$, ~ as an ordinary tensor field with
piecewise smooth transformations\footnote{Of course, some transformations in a
generalized sense might be defined also on the $\tilde\Omega(Ch_{1})\cup \tilde\Omega(Ch_{2})$ set.}.
\end{association}

\begin{proof}
All this immediately follows from what was done in
the previous proof, and from the fact that union of sets with Lebesgue measure 0 has Lebesgue measure 0.
\end{proof}

Note, that if there exists such point that for the object ~$\Gamma^{m}_{A}(M)$ we have

\begin{equation*}
\lim_{\epsilon\to 0}
T^{\mu...}_{\nu...}(\omega_{\epsilon})=\pm\infty
\end{equation*}
at that
point, then the field associated to this object, can be associated to
another object, which is nonequivalent to this object. This means that the same ~field can be
associated to mutually non-equivalent elements of ~$\Gamma^{m}_{A}(M)$. This is
explicitly shown and proven by the next example.

\newtheorem{dfunction}[association1]{Theorem}

\begin{dfunction}
 Take $\delta(q,Ch_{k}(\Omega_{Ch}))\in D'_{(\mathcal{\mathcal{\mathcal{\mathcal{S}}}} o)}(M)$ being defined as mapping from each 4-form ~$C^{P}_{S (\mathcal{\mathcal{\tilde S}})}(M)~(\mathcal{\mathcal{\mathcal{\mathcal{S}}}}\subseteq\mathcal{\mathcal{\tilde S}})$ ~to the
 value of this form's density at the point $q$ in the chart ~$Ch_{k}(\Omega_{Ch})\in\mathcal{\mathcal{\mathcal{\mathcal{\mathcal{S}}}}}$,~ ($q\in\Omega_{Ch}$). Then any
 power ~$n\in\mathbb{N}_{+}$ of $\delta(q,Ch_{k}(\Omega_{Ch}))$ is associated to the function being defined on the domain $M\setminus\{q\}$ and everywhere 0. Note that this function is associated to any power ($n\in\mathbb{N}_{+}$) (being a nonzero natural number) of $\delta(q)$, but different powers of $\delta(q)$ are mutually nonequivalent\footnote{It is hard to find in our theory a more ``natural'' definition generalizing the concept of delta function from $\mathbb{R}^{n}$. But there is still another natural generalization: it is an object from ~$\Gamma^{0}_{(\cup_{n}\mathcal{\mathcal{\mathcal{\mathcal{S}}}}_{n}o)}(M)$, defined as:~
$\delta(Ch_{k}(\Omega_{Ch}),q,\omega)=\omega'\left[(Ch_{k})(\tilde
q)\right]$,
~$Ch_{k}(\Omega_{Ch})\in\mathcal{\mathcal{\mathcal{\mathcal{\mathcal{S}}}}}_{n}$,
~$\omega\in C^{P}_{S (\mathcal{\tilde S}_{n})}(\Omega_{Ch})$
~($\mathcal{\mathcal{\mathcal{\mathcal{S}}}}_{n}\subset\mathcal{\mathcal{\tilde
S}}_{n}$), ~$\tilde q$ is image of $q$ given by the chart mapping
$Ch_{k}(\Omega_{Ch})$. So it gives value of the density $\omega'$ in
the chart $Ch_{k}(\Omega_{Ch})$, at the chart image of the point
$q$.}.\label{dfunction}
\end{dfunction}

\begin{proof}
Contracting powers of
~$\delta(q,Ch_{k}(\Omega_{Ch}))$ with a sequence of 4-forms from
arbitrary ~$A^{n}(Ch'(q',\Omega'_{Ch}),\mathcal{\mathcal{\tilde S}}), ~(q'\neq q)$ (they
have a support converging to another point than $q$) will give
0. For $q=q'$ (\ref{limit}) gives
\[\lim_{\epsilon\to 0}~\epsilon^{-4}~\omega'\left[(Ch_{k})(0)\right]=\pm\infty~.\]
 Now explore the equivalence between different powers of ~$\delta(Ch_{k},q)$.\\ $\delta^{n}(q,Ch_{k}(\Omega_{Ch}))$ applied to
~$\omega_{\epsilon}(x)\in A^{n}(Ch'_{k}(q,\Omega_{Ch}),\mathcal{\mathcal{\tilde S}})$ will
lead to the expression
~$\epsilon^{-4n}\omega^{n}(\frac{x}{\epsilon})$. Then if we
want to compute

\begin{equation*}
\lim_{\epsilon\to
0}\int_{Ch'(\Omega_{Ch})}\left(\frac{1}{\epsilon^{4n}}~\omega^{n}\left(\frac{\mathbf{x}}{\epsilon}\right)-\frac{1}{\epsilon^{4m}}~\omega^{m}\left(\frac{\mathbf{x}}{\epsilon}\right)\right)~\Phi(\mathbf{x})~d^{4}x
\end{equation*}
it leads to

\begin{equation*}
 \lim_{\epsilon\to
0}~\frac{1}{\epsilon^{4m-4}}~\Phi(0)~\int_{Ch'(\Omega_{Ch})}\left(\omega^{n}(\mathbf{x})-\frac{1}{\epsilon^{4(n-m)}}~\omega^{m}(\mathbf{x})\right)~d^{4}x
\end{equation*}

 which is for $n\neq m, ~~n,m\in\mathbb{N}_{+}$ clearly
divergent, hence nonzero.
\end{proof}

Note that despite of the fact that within our algebras we, naturally, have all the $n\in\mathbb{N_+}$ powers of the delta
distribution, they are for $n>1$, ~unfortunately, ~\emph{not} equivalent to any distribution.\newline

\newtheorem{nsmooth}[association1]{Theorem}

\begin{nsmooth}
 We see that the map $A_{s}$ is linear (in the sense analogous to \ref{linear}), and for arbitrary number of
~$g^{\mu...}_{\nu...},...,h^{\mu...}_{\nu...}\in\Lambda\cap\Gamma^{m}_{E (\cup_{n}At(\mathcal{\mathcal{\tilde S}}_{n})o)}(M)$ one has:
Take $\forall\Omega_{Ch}$, ~$\forall n$, ~$\forall \mathcal{\mathcal{\mathcal{\mathcal{S}}}}\subset\mathcal{\mathcal{\tilde S}}_{n}\cap At(\mathcal{\mathcal{\tilde S}})$, ~$\forall Ch_{k}(\Omega_{Ch})\in\mathcal{\mathcal{\mathcal{\mathcal{\mathcal{S}}}}}$,~ $\cup_{i}\tilde\Omega_{i}(Ch_{k})$~ to be the union of all $\tilde\Omega_{i}(Ch_{k})$~ related to the objects $g^{\mu...}_{\nu...},...,h^{\mu...}_{\nu...}$. Then
\[A_{s}(g^{\alpha...}_{\beta...}\otimes...\otimes
h^{\mu...}_{\nu...})=A_{s}(g^{\alpha...}_{\beta...})\otimes...\otimes
A_{s}(h^{\mu...}_{\nu...})\]
on $\Omega_{Ch}\setminus \cup_{i}\tilde\Omega_{i}(Ch_{k})$. ~Here the
first term is a product between $\Lambda$ objects and the second is the
classical tensor product.
\end{nsmooth}

\begin{proof}
It is trivially connected with previous proofs:
Note that from the definition (\ref{DefAssoc}) for appropriate 4-form fields $\omega_{\epsilon}$ we have
\begin{equation}
A_{s}(g^{\alpha...}_{\beta...}\otimes...\otimes
h^{\mu...}_{\nu...})(Ch_{k})=\lim_{\epsilon\to 0}~ g^{\alpha...}_{\beta...}(Ch_{k},\omega_{\epsilon})...
h^{\mu...}_{\nu...}(Ch_{k},\omega_{\epsilon}).
\end{equation}

But for the objects $g^{\alpha...}_{\beta...},...,h^{\alpha...}_{\beta...}\in\Lambda$, with respect to the theorem (\ref{Association1}) necessarily
\begin{equation}
\lim_{\epsilon\to 0}~ g^{\alpha...}_{\beta...}(Ch_{k},\omega_{\epsilon})...
h^{\mu...}_{\nu...}(Ch_{k},\omega_{\epsilon})=A_{s}(g^{\alpha...}_{\beta...})\otimes...\otimes
A_{s}(h^{\mu...}_{\nu...}).
\end{equation}
This proves the theorem.
\end{proof}

This means that the result of tensor multiplication of elements from $\Lambda$ (it has product of two scalars as a subcase) is always equivalent to some element from $\Lambda$.
This is a result closely related to the theorem (\ref{multiplication}). It tells us that multiplication is a mapping between equivalence classes formed of more constrained classes as those mentioned in the part $a)$ of the theorem (\ref{multiplication}).

\theoremstyle{definition}
\newtheorem{conj1}{Conjecture}[section]
\begin{conj1}
 If ~$T^{\mu...}_{\nu...}\in D'^{m}_{n (\mathcal{\mathcal{\mathcal{\mathcal{S}}}} o)}(M)$ has an
associated field, then it transforms on its domains as a
tensor field.
\end{conj1}

\theoremstyle{definition}
\newtheorem{conj2}[conj1]{Conjecture}
\begin{conj2}
 If ~$T^{\mu...}_{\nu...}\in \Gamma^{m}_{(\mathcal{\mathcal{\mathcal{\mathcal{S}}}} o)}(M)$ has an
associated field and ~$L^{\alpha...}_{\beta...}\in\Gamma^{n}_{S}(M)$,
then ~$A_{s}(T^{\mu...}_{\nu...}\otimes L^{\alpha...}_{\beta...})=A_{s}(T^{\mu...}_{\nu...})\otimes A_{s}(L^{\alpha...}_{\beta...})$. ~~(The $\otimes$ sign has again slightly different meaning on the different sides of
the equation).
\end{conj2}

\theoremstyle{definition}
\newtheorem{conj3}[conj1]{Conjecture}
\begin{conj3}
Take
~$C^{\mu...}_{\nu...},D^{\mu...}_{\nu...},F^{\alpha...}_{\beta...},B^{\alpha...}_{\beta...}\in\Gamma^{m}_{A}(M)$,
~such that they belong to the same classes
~$\Gamma^{m}_{(At(\mathcal{\mathcal{\tilde S}}))A}(M)$. Also assume
that ~$\forall\mathcal{\mathcal{\tilde S}}$, such that
~$C^{\mu...}_{\nu...},D^{\mu...}_{\nu...},F^{\alpha...}_{\beta...},\\
B^{\alpha...}_{\beta...}\in\Gamma^{m}_{(At(\mathcal{\mathcal{\tilde
S}}))A}(M)$, ~and ~$\forall
\mathcal{\mathcal{\mathcal{\mathcal{S}}}}\subseteq\mathcal{\mathcal{\tilde
S}}$,~ each of the elements
~$C^{\mu...}_{\nu...},D^{\mu...}_{\nu...},F^{\alpha...}_{\beta...},B^{\alpha...}_{\beta...}$
has for every ~$At(\mathcal{\mathcal{\tilde S}})$ associated fields
defined on the whole $M$. Then ~$F^{\alpha...}_{\beta...}\approx
B^{\alpha...}_{\beta...}$ and ~$C^{\mu...}_{\nu...}\approx
D^{\mu...}_{\nu...}$ implies ~$(C\otimes
F)^{\mu...\alpha...}_{\nu...\beta...}\approx (D\otimes
B)^{\mu...\alpha...}_{\nu...\beta...}$.
\end{conj3}

\subsection{Covariant derivative}

~~~~~The last missing fundamental concept is the
covariant derivative operator on generalized tensor fields (GTF). This operator is necessary to formulate an appropriate language for physics and generalize
physical laws. Such an operator must obviously reproduce our concept
of the covariant derivative on the smooth tensor fields (through the
given association relation to the smooth manifold). This is provided
in the following section. The beginning of the first part is again devoted to
fundamental definitions. The beginning of the second part gives us fundamental
theorems, again just generalizing Colombeau results for our case. After these theorems we, (similarly to previous section),
formulate conjectures representing the very important and natural
extensions of our results (bringing a lot of new significance to
our results). The last subsection in the second part being again called ``some additional theory''
brings (analogously to previous section) just physical insight to our abstract calculus and is of lower mathematical
importance.

\subsubsection{Definition of  $\partial$-derivative and connection coefficients}

\theoremstyle{definition}
\newtheorem{defn15}{Definition}[section]
\begin{defn15}
 We define a map, called the $\partial$-derivative, given by smooth vector
field $U^{i}$ (smooth in the atlas $\mathcal{\mathcal{\mathcal{\mathcal{\mathcal{S}}}}}$)~ as a mapping
$\Gamma^{m}_{(\mathcal{\mathcal{\mathcal{\mathcal{S}}}})}(M)\to \Gamma^{m}_{(\mathcal{\mathcal{\mathcal{\mathcal{S}}}})}(M)$,
given on its domain ~$\forall\Omega_{Ch}$~ and ~$Ch_{k}(\Omega_{Ch})$ as:
\begin{equation*}
T^{\mu...}_{\nu...~,(U)}(Ch_{k},\omega)\equiv
-T^{\mu...}_{\nu...}(Ch_{k},(U^{\alpha}\omega)_{,\alpha})~~~~~~\omega\in
C^{P}(\Omega_{Ch}).
\end{equation*}
Here
~$(U^{\alpha}\omega)_{,\alpha}$ is understood in the following way: We
express ~$U^{\alpha}$ in the chart ~~$Ch_{k}(\Omega_{Ch})$ ~~and take the derivatives
~~$\left(U^{\alpha}(Ch_{k})~\omega'(Ch_{k})\right)_{,\alpha}$~~ in the same chart\footnote{We will further express that the derivative is taken in $Ch_{k}(\Omega_{Ch})$ by using the notation $(U^{\alpha}\omega'(Ch_{k}))_{, [(Ch_{k})\alpha]}$.~} ~$Ch_{k}(\Omega_{Ch})$.~ They give us some
function in $Ch_{k}(\Omega_{Ch})$, which can be (in this chart) taken as expression for density of some object from
~$C^{P}(\Omega_{Ch})$. This means we trivially extend it to $M$ by taking it to be 0
everywhere outside ~$\Omega_{Ch}$.~ This is the object used as an argument in $T^{\mu...}_{\nu...}$.
\end{defn15}

To make a consistency check: This ~$T^{\mu...}_{\nu...,(U)}$~ is an
object which is defined at least on the domain ~$C^{P}_{S
(\mathcal{\tilde S})}(M)$~ for
~$\mathcal{\mathcal{\mathcal{\mathcal{S}}}}\subseteq\mathcal{\mathcal{\tilde
S}}$~ ($\mathcal{\mathcal{\mathcal{\mathcal{\mathcal{S}}}}}$
related to $U^{i}$ in the sense that $U^{i}$ is smooth in
$\mathcal{\mathcal{\mathcal{\mathcal{\mathcal{S}}}}}$),~ and is
continuous on the same domain. This means it belongs to the class
~$\Gamma^{n}(M)$.~ To show this take some arbitrary ~$\Omega_{Ch}$
and some arbitrary chart
~$Ch_{k}(\Omega_{Ch})\in\mathcal{\mathcal{\mathcal{\mathcal{\mathcal{S}}}}}$.
We see that within ~$Ch_{k}(\Omega_{Ch})$ the expression
~$(U^{\alpha}\omega'),_{\alpha}$ is smooth and describes 4-forms,
which are compactly supported, with their support being subset of
$\Omega_{Ch}$. Hence they are from the domain of
~$T^{\mu...}_{\nu...}$ in every chart from
$\mathcal{\mathcal{\mathcal{\mathcal{\mathcal{S}}}}}$. In any
arbitrary chart from
$\mathcal{\mathcal{\mathcal{\mathcal{\mathcal{S}}}}}$ we
trivially observe, (from the theory of distributions), that if
~$\omega_{n}\to\omega$, than ~$T^{\mu...}_{\nu...,
(U)}(\omega_{n})\to T^{\mu...}_{\nu...,(U)}(\omega)$. ~It means that
$T^{\mu...}_{\nu...,(U)}(\omega)$ is continuous.

\theoremstyle{definition}
\newtheorem{defn16}[defn15]{Definition}
\begin{defn16}
 Now, by generalized connection we denote an
object from ~$\Gamma^{3}(M)$ such that:
\begin{itemize}
\item
The set
\[\cup_{\Omega_{Ch}}\cup_{\mathcal{\mathcal{\tilde S}}}\bigg\{\big(\omega,Ch(\Omega_{Ch})\big):\omega\in C^{P}_{S(\mathcal{\mathcal{\tilde S}})}(\Omega_{Ch}), ~~Ch(\Omega_{Ch})\in\mathcal{A}\bigg\}\]
 belongs to its domain.
\item
It is ~$\forall\mathcal{\mathcal{\tilde S}}$,~ $\forall\Omega_{Ch}$,~ $\forall Ch_{k}(\Omega_{Ch})\in\mathcal{A}$~
continuous ~on ~$C^{P}_{S(\mathcal{\mathcal{\tilde S}})}(\Omega_{Ch})$ ~with ~$Ch_{k}(\Omega_{Ch})$ ~taken as its argument.
\item
It transforms as:
\begin{equation}
\Gamma^{\alpha}_{\beta\gamma}(Ch_{2},\omega)=
\Gamma^{\mu}_{\nu\delta}(Ch_{1},((J^{-1})^{\nu}_{\beta}(J^{-1})^{\delta}_{\gamma}J^{\alpha}_{\mu}-J^{\alpha}_{m}(J^{-1})^{m}_{\beta,\gamma})~\omega).
\end{equation}
\end{itemize}
\end{defn16}

\subsubsection{Definition of covariant derivative}

\theoremstyle{definition}
\newtheorem{defn17}[defn15]{Definition}
\begin{defn17}
 By a covariant derivative ~ (on $\Omega_{Ch}$) in the
direction of a vector field ~$U^{i}(M)$, smooth with respect to atlas $\mathcal{\mathcal{\mathcal{\mathcal{\mathcal{S}}}}}$, of an object
~$T^{\mu...}_{\nu...}\in \Gamma^{m}_{(\mathcal{\mathcal{\mathcal{\mathcal{S}}}})}(\Omega_{Ch})$,~ we mean:
\begin{equation}
D_{C (U)}T^{\mu...}_{\nu...}(\omega)\equiv
T^{\mu...}_{\nu...~,(U)}(\omega)+\Gamma^{\mu}_{\alpha\rho}(\omega)T^{\rho...}_{\nu...}(U^{\alpha}\omega)-\Gamma^{\alpha}_{\nu\rho}(\omega)T^{\mu...}_{\alpha...}(U^{\rho}\omega).\label{cov}
\end{equation}

\end{defn17}

\medskip

This definition (\ref{cov}) automatically defines covariant derivative everywhere on
~$\Gamma^{m}_{(\mathcal{\mathcal{\mathcal{\mathcal{S}}}})}(M)$. This can be easily observed: The
$\partial$-derivative still gives us an object from ~$\Gamma^{m}_{(\mathcal{\mathcal{\mathcal{\mathcal{S}}}})}(M)$, and the
second term containing generalized connection is from
~$\Gamma^{m}_{(\mathcal{\mathcal{\mathcal{\mathcal{S}}}})}(M)$ trivially too.

\theoremstyle{definition}
\newtheorem{defn177}[defn15]{Definition}
\begin{defn177}
Furthermore, let us
extend the definition of covariant derivative to the class ~$\Gamma^{m}_{
(\mathcal{\mathcal{\mathcal{\mathcal{S}}}})A}(M)$ just by stating that on every nonlinear object (note that every such object is constructed by tensor product of linear objects) it is defined by the Leibniz rule. (This means
it is a standard derivative operator, since it is trivially linear
as well.)
\end{defn177}

\subsubsection{The $S'_n$ class}

\paragraph{Notation.}
Take as ~$\mathcal{D}_{n}$ some $n$-times continuously
differentiable subatlas of $\mathcal{A}$. Then the class $S'_{n}$ related to the atlas ~$\mathcal{D}_{n}$~
is formed by objects ~$T^{\mu...}_{\nu...}\in\Gamma^{a}_{E
(\cup_{m}At(\mathcal{\mathcal{\tilde S}}_{m})o)A}(M)$, such that:
\begin{itemize}
\item
$\cup_{m}\mathcal{\mathcal{\tilde S}}_{m}=\mathcal{D}_{n}$~ and
~$\forall m,~~At(\mathcal{\mathcal{\tilde
S}}_{m})=\mathcal{D}_{n}$,
\item
it is given  $\forall\Omega_{Ch}$, ~$\forall
Ch_{k}(\Omega_{Ch})\in\mathcal{D}_{n}$, ~$\forall m$, ~$\forall\omega\in
C^{P}_{S (\mathcal{\mathcal{\tilde S}}_{m})}(\Omega_{Ch})$ as a map
\[(\omega, Ch_{k})\to\int_{\Omega_{Ch}}T^{\mu...}_{\nu...}(Ch_{k})~\omega.~~~~~~~\]

 Here $T^{\mu...}_{\nu...}(Ch_{k})$ is a multi-index matrix of $n$-times continuously differentiable functions (if it is being expressed in any arbitrary chart from $\mathcal{D}_{n}$).
\end{itemize}

\subsubsection{Basic equivalence relations related to differentiation and some of the
interesting conjectures}

\newtheorem{lcovder1}{Theorem}[section]

\begin{lcovder1}
The following statements hold:
\begin{itemize}
\item[a)] Take a vector field ~$U^{i}$, which is smooth at
 ~$\mathcal{\mathcal{\mathcal{\mathcal{S}}}}\subset \mathcal{D}_{n+1}\subset \mathcal{D}_{n}$ for $n\geq 1$, ~(formally including also ~$n=n+1=\infty$, hence $\mathcal{\mathcal{\mathcal{\mathcal{\mathcal{S}}}}}$), plus a generalized
 connection
 ~$\Gamma^{\mu}_{\nu\alpha}\in\Gamma^{3}_{E}(M)$,~ and
~$T^{\mu...}_{\nu...}\in D'^{a}_{b}(M)\cap S'_{n}$. ~ $S'_{n}$ is
here related to the given ~$\mathcal{D}_{n}$. Then it follows that:
$D_{C(U)}T^{\mu...}_{\nu...}$ is an object from $S'_{n+1}$ ~(being
related to the given $\mathcal{D}_{n+1}$). Moreover, the equivalence
class ~$\tilde S'_{n+1}$ of the image contains exactly one linear
element given for every chart from its domain as integral from some
multi-index matrix of piecewise continuous functions. Particularly
this element is for $\forall\Omega_{Ch}$ and arbitrary such
$\omega\in C^{P}(\Omega_{Ch})$, ~$Ch_{k}(\Omega_{Ch})$ that are from
its domain a map:
\begin{equation}
(Ch_{k},\omega)\to\int
U^{\alpha}T^{\mu...}_{\nu...;\alpha}(Ch_{k})~\omega.
\end{equation}
 ~(Here ``$;$'' means the
classical covariant derivative related to the ``classical'' connection, components of which are in $Ch_{k}$ given by
~$\Gamma^{\alpha}_{\beta\delta}(Ch_{k})$, and the tensor field, components of which are in $Ch_{k}$ given by $T^{\mu...}_{\nu...}(Ch_{k})$.)
\item[b)] Take $U^{i}$ being smooth in $\mathcal{\mathcal{\mathcal{\mathcal{\mathcal{S}}}}}$ and
~$\Gamma^{\mu}_{\nu\alpha}\in\Gamma^{3}_{S}(M)$. Then the following holds:
The covariant derivative is a map  ~$D'^{m}_{n (\mathcal{\mathcal{\mathcal{\mathcal{S}}}}
o)}(M)\to\tilde\Gamma^{m+n}_{ (\mathcal{\mathcal{\mathcal{\mathcal{S}}}} o)A}(M)$~ and the classes
~~$\tilde\Gamma^{m+n}_{ (\mathcal{\mathcal{\mathcal{\mathcal{S}}}} o)A}(M)$ ~of the image contain (exactly)
one element of $D'^{m}_{n (\mathcal{\mathcal{\mathcal{\mathcal{S}}}} o)}(M)$.
\end{itemize}\label{lcovder1}
\end{lcovder1}

\begin{proof}

\begin{itemize}
\item[a)]Covariant derivative is, on its domain, given as
\begin{equation*}
(\omega, Ch_{k})\to T^{\mu...}_{\nu...~,(U)}(\omega)+\Gamma^{\mu}_{\alpha\rho}(\omega)T^{\rho...}_{\nu...}(U^{\alpha}\omega)-\Gamma^{\alpha}_{\nu\rho}(\omega)T^{\mu...}_{\alpha...}(U^{\rho}\omega).
\end{equation*}

Take the first term ~$T^{\mu...}_{\nu...,(U)}$. Express it on
$\Omega_{Ch}$ in ~$Ch_{k}(\Omega_{Ch})\in  \mathcal{D}_{n}$ as the
map:
\[\omega\to\int_{Ch_{k}(\Omega_{Ch})}T^{\mu....}_{\nu...}(Ch_{k})~\omega'(Ch_{k})~d^{4}x.\]
 ~So analogously:
\[T^{\mu...}_{\nu...,(U)}(Ch_{k},\omega)=
-\int_{Ch_{k}(\Omega_{Ch})}T^{\mu...}_{\nu...}(Ch_{k})\left(U^{\alpha}(Ch_{k})~\omega'(Ch_{k})\right)_{,[(Ch_{k})\alpha]}~d^{4}x.\]

Here $\omega'$ is in arbitrary chart from ~$\mathcal{D}_{n}$,
being~ $n$-times continuously differentiable, (the domain is limited
to such objects by the second covariant derivative term, which is
added to the $\partial$-derivative), and such that the expression
~$(U^{i}(Ch_{k})~\omega'(Ch_{k}))_{,[(Ch_{k})i]}$ is in any
$Ch_{k}(\Omega_{Ch})\in\mathcal{D}_{n}$~ $n$-times continuously
differentiable.

Now by using integration by parts,
(since all the objects under the integral are at least continuously
differentiable ($n\geq 1$), it can safely be used), and considering the compactness of
support we obtain:

\begin{eqnarray}
T^{\mu...}_{\nu...,(U)}(Ch_{k},\omega)~~~~~~~~~~~~~~~~~~~~~~~~~~~~~~~~~~~~~~~~~~~~~~~~~~~~~~~~~~~~~~~~~~~~~~~~~~~~~~~~~\nonumber\\
=-\int_{Ch_{k}(\Omega_{Ch})}T^{\mu...}_{\nu...}(Ch_{k})\left(U^{\alpha}(Ch_{k})~\omega'(Ch_{k})\right)_{,[(Ch_{k})\alpha]}~d^{4}x\nonumber\\ =\int_{Ch_{k}(\Omega_{Ch})}T^{\mu...}_{\nu...}(Ch_{k})_{,[(Ch_{k})\alpha]}U^{\alpha}(Ch_{k})~\omega'(Ch_{k})~d^{4}x.~~~
\end{eqnarray}
Then it holds that:
\begin{eqnarray}
T^{\mu...}_{\nu...,(U)}(Ch_{m}(\Omega_{Ch}),\omega)~~~~~~~~~~~~~~~~~~~~~~~~~~~~~~~~~~~~~~~~~~~~~~~~~~~~~~~~~~~~~~~~~~~~~~~~~\nonumber\\
=\int_{Ch_{m}(\Omega_{Ch})}T^{\mu...}_{\nu...}(Ch_{m})_{,[(Ch_{m})\alpha]}(U^{\alpha}(Ch_{m})~\omega'(Ch_{m})~d^{4}x~~~~~~~~~~~~~~~~~~~~\nonumber\\
=\int_{Ch_{k}(\Omega_{Ch})}(J^{\mu}_{\beta}(J^{-1})
^{\delta}_{\nu}....T^{\delta...}_{\beta...}(Ch_{k}))_{,[(Ch_{k})\alpha]}U^{\alpha}(Ch_{k})~\omega'(Ch_{k})~d^{4}x.~~~~~
\end{eqnarray}
We see that this is defined and continuous in any arbitrary chart
from $\mathcal{D}_{n+1}$ for every $C^{P}_{S
(\mathcal{\mathcal{\tilde S}})}(M)$, such that $\exists
\mathcal{\mathcal{\mathcal{\mathcal{S}}}}\subset\mathcal{\mathcal{\tilde
S}}$ and $\mathcal{\mathcal{\mathcal{\mathcal{S}}}}\subset
\mathcal{D}_{n+1}$.

Now we see that the second term in the covariant
derivative expression is equivalent to the map (with the same domain):
\begin{equation*}
(Ch_{k},\omega)\to\int(\Gamma^{\mu}_{\alpha\rho}U^{\alpha}T^{\rho...}_{\nu...}-\Gamma^{\alpha}_{\nu\rho}U^{\rho}T^{\mu...}_{\alpha...})(Ch_{k})~\omega,
\end{equation*}

(see
theorem \ref{multiplication}), and between
charts the objects appearing
inside the integral transform exactly as their classical analogues. This must hold, since $T^{\mu...}_{\nu...}$ is
everywhere continuous (in every chart considered), hence on every
compact set bounded, so the given object is well defined.
This means that when we fix this object in chart ~$Ch_{m}(\Omega_{Ch})$, and express it through the chart ~$Ch_{k}(\Omega_{Ch})$ and
Jacobians (with the integral expressed at chart
~$Ch_{k}(\Omega_{Ch})$ ), as in the previous case, we discover
(exactly as in the classical case), that the resulting object under
the integral transforms as some object ~$D'^{m}_{n E}(M)$, with the classical expression for the
covariant derivative of a tensor field appearing under the integral.
\item[b)]The resulting object is defined particularly only on ~$C^{P}_{S
(\mathcal{\mathcal{\tilde S}})}(M)$ ~($\mathcal{\mathcal{\mathcal{\mathcal{S}}}}\subset\mathcal{\mathcal{\tilde S}}$). We have to realize that
~$T^{\mu...}_{\nu...}$ can be written as a ($N\to\infty$) weak limit
(in every chart from $\mathcal{\mathcal{\mathcal{\mathcal{\mathcal{S}}}}}$) of ~$T^{\mu...}_{\nu... N}\in D'^{m}_{n S
(\mathcal{\mathcal{\mathcal{\mathcal{S}}}} o)}(M)$. It is an immediate result of previous constructions and
Colombeau theory, that

\begin{equation}
-\Gamma^{\alpha}_{\nu\rho}(.)~T^{\mu...}_{\alpha...}(U^{\rho}.)\approx -T^{\mu...}_{\alpha...}(\Gamma^{\alpha}_{\nu\rho}U^{\rho}.)
\label{object}
\end{equation}

Now take ~$\forall\Omega_{Ch}$~ both, $T^{\mu...}_{\nu...,(U)}$ and (\ref{object}) fixed in arbitrary ~$Ch_{k}(\Omega_{Ch})\in\mathcal{\mathcal{\mathcal{\mathcal{\mathcal{S}}}}}$.~ Write both of those objects as limits of integrals of some
sequence of ``smooth'' objects in ~$Ch_{k}(\Omega_{Ch})$.~ But now we
can again use for ~$T^{\mu...}_{\nu...,(U)}$~ an integration per parts and from the ``old'' tensorial relations; we get the
``tensorial'' transformation properties under the limit. This means
that the resulting object, which is a limit of those objects
transforms in the way the ~$D'^{m}_{n}(M)$ objects transform.
\end{itemize}
\end{proof}

\newtheorem {meantime}[lcovder1]{Theorem}
\begin{meantime}

Part $a)$ of the theorem (\ref{lcovder1}) can be also formulated
through a generalized concept of covariant derivative, where we do
not require the $U^{i}$ vector field to be smooth at some
$\mathcal{\mathcal{\mathcal{\mathcal{\mathcal{S}}}}}$, but it
is enough if it is $n+1$ differentiable in $\mathcal{D}_{n+1}$.

\end{meantime}

\begin{proof}
We just have to follow our proof and realize that the only reason why we used
smoothness of $U^{i}$ in $\mathcal{\mathcal{\mathcal{\mathcal{\mathcal{S}}}}}$ was that it is required by our
definition of covariant derivative (for another good reasons related to different cases).
\end{proof}

This statement has a crucial importance, since it shows that not only
all the classical calculus of smooth tensor fields with all the
basic operations is contained in our language (if we take the
equivalence instead of equality being the crucial part of our
theory), but it can be even extended to arbitrary objects from $S'_{n}$.
(If the covariant derivative is obtained through connection from the class
$\Gamma^{3}_{E}(M)$.) In other words it is more general than the
classical tensor calculus.

We can now think
about conjectures extending our results in a very important way:

\theoremstyle{definition}
\newtheorem{cc1}{Conjecture}[section]
\begin{cc1}
 Take an arbitrary piecewise smooth, and on every compact set bounded \footnote{To be exact, the expression ``covariant derivative'' is used in this and the following conjecture in a more general way, since we do not put on $U^{i}$ the condition of being smooth in some subatlas $\mathcal{\mathcal{\mathcal{\mathcal{\mathcal{S}}}}}$.} vector field ~$U^{i}$. Take also ~$\Gamma^{\mu}_{\nu\alpha}\in\Gamma^{3}_{E}(M)$ and such
~$T^{\mu...}_{\nu...}\in D'^{m}_{n E}(M)$, that ~$\forall\Omega_{Ch}$,~ $\forall Ch_{k}(\Omega_{Ch})\in\mathcal{A}$ ~$\exists Ch_{l}(\Omega_{Ch})\in\mathcal{A}$,~ in which\footnote{This means we are trivially integrating~ $T^{\mu...}_{\nu...}\Gamma^{\alpha}_{\beta\delta}$~on compact sets within subset of $\mathbb{R}^{4}$ given as image of the given chart mapping. }
\[\int_{Ch_{l|\Omega'}(\Omega_{Ch})}
T^{\mu...}_{\nu...}(Ch_{k})\Gamma^{\alpha}_{\beta\delta}(Ch_{k})~d^{4}x\]
converges on every compact set~ $\Omega'\subset\Omega_{Ch}$. ~Then the following holds:
The covariant derivative (along $U^{i}$) maps this object to an
element of some equivalence class from ~$\tilde\Gamma^{m+n}_{ A }(M)$. This class contains (exactly) one
element from ~$D'^{m}_{n}(M)$.
\end{cc1}

\theoremstyle{definition}
\newtheorem{cc2}[cc1]{Conjecture}
\begin{cc2}
Take ~$U^{i}$ being a piecewise smooth vector field, and
~$\Gamma^{\mu}_{\nu\alpha}\in\Gamma^{3}_{S}(M)$. Then the following
holds: The covariant derivative along this vector field is a map:
~$D'^{m}_{n (\cup_{l}\mathcal{\mathcal{\mathcal{\mathcal{S}}}}_{l}
~o)}(M)\to\tilde\Gamma^{m+n}_{
(\cup_{l}\mathcal{\mathcal{\mathcal{\mathcal{S}}}}_{l}~o)A}(M)$, and
the classes ~$\tilde\Gamma^{m+n}_{
(\cup_{l}\mathcal{\mathcal{\mathcal{\mathcal{S}}}}_{l}~o)A}(M)$ of
the image contain (exactly) one element of ~$D'^{m}_{n}(M)$.
\end{cc2}

\newtheorem{Comm}[lcovder1]{Theorem}
\begin{Comm}
For ~$U^{i}$ being a smooth tensor field in $\mathcal{\mathcal{\mathcal{\mathcal{\mathcal{S}}}}}$ with the connection taken
from ~$\Gamma^{3}_{S}(M)$, ~$T^{\mu...}_{\nu...}\in\Gamma^{m}_{ (\mathcal{\mathcal{\mathcal{\mathcal{S}}}}
o)A}(M)$ and ~$B^{\mu...}_{\nu...}\in\Gamma^{m}_{(\mathcal{\mathcal{\mathcal{\mathcal{S}}}} o)}(M)$, it holds that
~$T^{\mu...}_{\nu...}\approx B^{\mu...}_{\nu...}$~ implies
~$D^{n}_{C(U)}T^{\mu...}_{\nu...}\approx D^{n}_{C(U)}B^{\mu...}_{\nu...}$~ for arbitrary natural number
$n$.

\end{Comm}

\begin{proof}
Pick an arbitrary ~$\Omega_{Ch}$ and an arbitrary fixed chart
~$Ch'(\Omega_{Ch})\in\mathcal{\mathcal{\mathcal{\mathcal{\mathcal{S}}}}}$. Such a chart maps all the 4-forms from
the domain of ~$\Gamma^{m}_{ (\mathcal{\mathcal{\mathcal{\mathcal{S}}}}o)A}(M)$ objects to smooth compact
supported functions (given by densities expressed in that chart). The objects
~$T^{\mu...}_{\nu...}((U^{\alpha}\omega)_{,\alpha})$ and
~$B^{\mu...}_{\nu...}((U^{\alpha}\omega)_{,\alpha})$  are taken as objects of the
Colombeau algebra (the connection, fixed in that chart,
is also an  object of the Colombeau algebra) and are equivalent to
~$U^{\alpha}(\omega)T^{\mu...}_{\nu...,\alpha}(\omega)$ and
~$U^{\alpha}(\omega)B^{\mu...}_{\nu...,\alpha}(\omega)$. Here the
derivative means the ''distributional derivative'' as used in the
Colombeau theory (fulfilling the Leibniz rule) and ~$U^{\alpha}(\omega)$
is simply a ~$D'^{m}_{n S (\mathcal{\mathcal{\mathcal{\mathcal{S}}}} o)}(M)$ object with the given vector
field appearing under the integral. But in the Colombeau
theory one knows that if some object is equivalent to a distributional
object, then their derivatives of arbitrary degree are also equivalent.
It also holds that if any arbitrary object is equivalent to a
distributional object, then they remain equivalent after being multiplied by arbitrary smooth
distribution. In the fixed chart we have (still in the Colombeau theory sense),

\begin{equation*}
T^{\mu...}_{\nu...}\approx B^{\mu...}_{\nu...}.
\end{equation*}

 But since their
$\partial$-derivatives were, in the fixed chart, obtained only by the distributional derivatives and multiplication by a smooth function, also their
$\partial$-derivatives must remain equivalent. The same holds about the second
covariant derivative term (containing connection). So the objects
from classical Colombeau theory, (classical theory just trivially
extended to what we call multi-index matrices of functions),
obtained by the chart mapping of the covariant derivatives of ~$T^{\mu...}_{\nu...}$~ and ~$B^{\mu...}_{\nu...}$,~ are equivalent in the
sense of the Colombeau theory. But the ~$\Omega_{Ch}$ set was
arbitrary and also the chart was an arbitrary chart from the
domain of ~$T^{\mu...}_{\nu...},~ B^{\mu...}_{\nu...}$.~ So ~$T^{\mu...}_{\nu...}$~ and ~$B^{\mu...}_{\nu...}$~ are equivalent with respect to our definition.
\end{proof}

We can try to extend this statement to a conjecture:

\theoremstyle{definition}
\newtheorem{ADD3}[cc1]{Conjecture}
\begin{ADD3}
Take ~$U^{i}$ being piecewise smooth tensor field, take connection
from the class ~$\Gamma^{3}_{S}(M)$,
~$T^{\mu...}_{\nu...}\in\Gamma^{m}_{
(\cup_{l}\mathcal{\mathcal{\mathcal{\mathcal{S}}}}_{l}~o)A}(M)$, and
~$B^{\mu...}_{\nu...}\in\Gamma^{m}_{(\cup_{l}\mathcal{\mathcal{\mathcal{\mathcal{S}}}}_{l}~o)}(M)$.
Then it holds that ~$T^{\mu...}_{\nu...}\approx
B^{\mu...}_{\nu...}$~ implies
~$D^{n}_{C(U)}T^{\mu...}_{\nu...}\approx
D^{n}_{C(U)}B^{\mu...}_{\nu...}$~ for arbitrary natural number $n$,
if such covariant derivative exists.
\end{ADD3}

This conjecture in fact means that if we have connection from the
class $\Gamma^{3}_{S}(M)$, then the covariant derivative is a map
from such element of the class ~$\tilde\Gamma^{m}_{
(\cup_{l}\mathcal{\mathcal{\mathcal{\mathcal{S}}}}_{l}~o)A}(M)$,
that it contains some linear element, to ~$\tilde\Gamma^{m}_{A}(M)$.
Note that we can also try to prove an extended version of the
conjecture, taking the same statement and just extending the classes
~$\Gamma^{m}_{
(\cup_{l}\mathcal{\mathcal{\mathcal{\mathcal{S}}}}_{l}~o)A}(M)$,
~$\Gamma^{m}_{
(\cup_{l}\mathcal{\mathcal{\mathcal{\mathcal{S}}}}_{l}~o)}(M)$~ to
the classes ~$\Gamma^{m}_{
(\cup_{l}\mathcal{\mathcal{\mathcal{\mathcal{S}}}}_{l})A}(M)$,
~$\Gamma^{m}_{
(\cup_{l}\mathcal{\mathcal{\mathcal{\mathcal{S}}}}_{l})}(M)$.

\subsubsection{Some additional theory}

\newtheorem{lcovder}[lcovder1]{Theorem}

\begin{lcovder}
Take ~$U^{i}$ to be vector field smooth in some
~$\mathcal{\mathcal{\mathcal{\mathcal{S}}}}\subset
\mathcal{D}_{n}$,~ with $\Gamma^{\mu}_{\nu\alpha}\in\Lambda$~ and
~$T^{\mu...}_{\nu...}\in S'_{n}\cap D'^{a}_{b}(M)$ ($n\geq 1$,
$S'_{n}$ is related to $\mathcal{D}_{n}$).~ Then $D_{C (U)}
T^{\mu...}_{\nu...}$ has an associated field which is on
~$M\setminus\tilde\Omega(Ch)$~ the classical covariant derivative of
~$A_{s}(T^{\mu...}_{\nu...})$. ($\tilde\Omega(Ch)$ is a set on which is
~$\Gamma^{\mu}_{\nu\alpha}(Ch)$~ continuous and is of 0 Lebesgue
measure.) It is defined on the whole $\mathcal{D}_{n+1}$
~($\mathcal{\mathcal{\mathcal{\mathcal{S}}}}\subset
\mathcal{D}_{n+1}$). That means association and covariant
differentiation in this case commute.\label{T244}
\end{lcovder}

\begin{proof}
Just take the definition of the classical covariant derivative, and
define the linear mapping given ~$\forall\Omega_{Ch}$~ and arbitrary
~$Ch_{k}(\Omega_{Ch})\in \mathcal{D}_{n}$~ as
\begin{equation}
(\omega, Ch_{k})\to\int_{\Omega_{Ch}}U^{\nu
}(Ch_{k})\left[A_{s}(T^{\mu...}_{\nu...})_{;\mu}\right](Ch_{k})~\omega.\label{AsC}
\end{equation}

This is an internally consistent definition, since
$\left[A_{s}(T^{\mu...}_{\nu...})_{;\mu}\right](Ch_{k})$ is defined everywhere
apart of a set having L measure 0. Now from our previous results
follows that everywhere outside ~$\tilde\Omega(Ch)$
\[\left[A_{s}(T^{\mu...}_{\nu...})_{;\mu}\right](Ch_{k})=T^{\mu...}_{\nu...
~;\mu}(Ch_{k}),\] and so the linear mapping (\ref{AsC}) is
equivalent to the object $D_{C(U)}T^{\mu...}_{\nu...}$. Then $U^{\nu
}A_{s}(T^{\mu...}_{\nu...})_{;\mu}=A_{s}(D_{C(U)}T^{\mu...}_{\nu...})$~everywhere
outside the set $\tilde\Omega(Ch)$.
\end{proof}

\newtheorem{meantime2}[lcovder1]{Theorem}

\begin{meantime2}
An extended analogy of theorem (\ref{T244}), can be proven, if we
use generalized concept of covariant derivative, without assuming
that vector field is smooth in some
$\mathcal{\mathcal{\mathcal{\mathcal{\mathcal{S}}}}}$, but only
$n+1$ continuously differentiable within $\mathcal{D}_{n+1}$.
\end{meantime2}

\begin{proof}
Exactly the same as before.
\end{proof}

This means that the aim to define a
concept of covariant derivative, ``lifted'' from the smooth manifold
and smooth tensor algebra to GTF in sense of association, has been achieved. It
completes the required connection with the old tensor
calculus.

\theoremstyle{definition}
\newtheorem{conj4}[cc1]{Conjecture}
\begin{conj4}
 Take ~$U^{i}$ smooth
in ~$\mathcal{\mathcal{\mathcal{\mathcal{\mathcal{S}}}}}$, ~~$T^{\mu...}_{\nu...}\in D'^{m}_{n (\mathcal{\mathcal{\mathcal{\mathcal{S}}}} o)}(M)$, ~such that ~$\exists ~A_{s}(T^{\mu...}_{\nu...})$,
~$\Gamma^{\mu}_{\nu\alpha}\in\Gamma^{3}_{S}(M)$. Then
~$\exists ~A_{s}(D_{C(U)}T^{\mu...}_{\nu...})$ and holds that \[A_{s}(D_{C(U)}T^{\mu...}_{\nu...})=U^{\alpha}A_{s}(T^{\mu...}_{\nu...})_{;\alpha}.\]
(Hence, similarly to $\otimes$, the covariant derivative
operator commutes with association for some significant number of objects.)
\end{conj4}

Note that for
every class ~$\Gamma^{m}_{At(\mathcal{\mathcal{\mathcal{\mathcal{S}}}})}(M)$ we can easily define the operator of Lie derivative along arbitrary in $\mathcal{\mathcal{\mathcal{\mathcal{\mathcal{S}}}}}$ smooth vector field $V$ (not along the generalized vector field, but even in the case of covariant derivative we did not prove
anything about larger classes of vector fields than smooth vector
fields). Lie derivative can be defined as
~$(L_{V}T,\omega)\equiv (T,L_{V}\omega)$. This is because the Lie derivative preserves $n$-forms and also preserves the
properties of such ~$C^{P}_{(\mathcal{\mathcal{\tilde S}})}(M)$ classes, for which it holds that ~$\mathcal{\mathcal{\mathcal{\mathcal{S}}}}\subset\mathcal{\mathcal{\tilde S}}$.

\subsection{Basic discussion of previous results and open questions}

We have constructed the algebra of GTFs, being able to incorporate
the concept of covariant derivative, (with the given conditions on
vector fields and connection), for a set of algebras constructed
from specific distributional objects. The use of these ideas in
physics is meaningful where the operations of tensor product and
covariant derivative give a map from appropriate subclass of
~$D'^{m}_{n}(M)$ class to the elements of ~$\tilde\Gamma^{m}(M)$
containing a ~$D'^{m}_{n}(M)$ element. This is always guaranteed to
work between appropriate subclasses of piecewise continuous
distributional objects, but a given physical equivalence might
specify a larger set of objects for which these operations provide
such mapping. Note that the whole problem lies in the multiplication
of distributions outside the ~$D'^{m}_{n E}(M)$ class. (For instance
it can be easily seen that square of ~$\delta(Ch_{k},q)$ as
introduced before is not equivalent to any distribution.) This is
because the product does not have to be necessarily equivalent to a
distributional object. Even worse, in case it is not equivalent to a
distributional object, the product is not necessarily a mapping
between equivalence classes of the given algebra elements ($\tilde
D'^{m}_{n A}(M)$). The same holds about contraction.

But even in such cases there can be a further hope. For example we
can abandon the requirement that certain quantities must be linear,
(for example the connection), and only some results of their
multiplication are really physical (meaning linear). Then it is a
question whether they should be constructed (constructed from the
linear objects as for example metric connection from the metric
tensor) through the exact equality or only through the equivalence.
If we take only the weaker (equivalence) condition, then there is a
vast number of objects we can choose, and many other important
questions can be posed. Even in the case that the mathematical
operations do depend on particular representatives of the
equivalence classes, there is no necessity to give up; in such
situation it might be an interesting question if there are any
specific ``paths'' which can be used to solve the physical
equivalence relations. The other point is that if these operations do depend
on the class members, then we can reverse this process. It means
that for example in the case of multiplication of two delta
functions we can find their nonlinear equivalents first and then
take their square, thus obtaining possibly an object belonging to an
equivalence class of a distribution.

As I mentioned in the introduction, these are not attempts to deal
with physical problems in a random, ad hoc way. Rather I want to
give the following interpretation to what is happening: The
differential equations in physics should be changed into
equivalence relations. For that reason they have plentiful solutions in the
given algebra\footnote{After one for example proves that covariant
derivative is a well defined operator on all GTF elements, then it
is the whole GTF algebra.}. (By a ``solution'' one here means any
object fulfilling the given relations.) One obtains much ``more''
solutions than in the case of classical partial differential
equations (but all the smooth distributions representing
``classical'' solutions of the ``classical'' initial value problem
are there), but what is under question is the possibility to
formulate the initial value problem for larger classes of objects
than $D'^{m}_{n E}(M)$ and $D'^{m}_{n (\mathcal{\mathcal{\mathcal{\mathcal{S}}}}o)}(M)$ (see the next
section). Moreover, if this is possible then there remains another
question about the physical meaning of those solutions. It means
that even in the case we get nonlinear objects as solutions of some
general initial value problem formulation, this does not have to be
necessarily something surprising; the case where physical laws are
solved also by physically meaningless solutions is nothing new. The
set of objects where we can typically search for physically
meaningful solutions is defined by most of the distributional
mappings (that is why classical calculus is so successful), but they
do not have to be necessarily the only ones.

\subsection{Some notes on the initial value problem within the partial
differential equivalence relations ($\approx$) on $D'^{m}_{n}(M)$}

In this section we will suggest how to complete the mathematical structure
developed and will get some idea how a physical problem can be
formulated in our language. It is again divided
into what we call ``basic ideas'' and ``some additional ideas''. The
first part is of a considerable importance, the second part is less
important, it just gives a suggestion how to recover the classical
geometric concept of geodesics in our theory.

\subsubsection{The basic ideas}

\paragraph{The approach giving the definition of the initial value problem}

Take a hypersurface (this can be obviously generalized to any
submanifold of lower dimension) ~$N\subset M$, which is such that it
gives in some subatlas
~$\mathcal{\mathcal{A}}_{N}\subset\mathcal{A}$~ a piecewise
smooth submanifold.  In the same time $\mathcal{\mathcal{A}}_{N}$
is such atlas that ~$\exists
\mathcal{\mathcal{\mathcal{\mathcal{S}}}}$~
~$\mathcal{\mathcal{\mathcal{\mathcal{S}}}}\subseteq
\mathcal{\mathcal{A}}_{N}$.

If we
consider space of 3-form fields living on $N$
(we give up on the
idea of relating them to 4-forms on $M$), we get two types of
important maps:

\begin{itemize}

\item
Take such ~$D'^{m}_{n E}(M)$ objects, that they have in every chart
from ~$\mathcal{\mathcal{A}}_{N}$ associated (tensor) fields
defined everywhere on $N$, apart from a set having 3 dimensional
Lebesgue measure equal to 0. Such objects can be, in every smooth
subatlas
~$\mathcal{\mathcal{\mathcal{\mathcal{S}}}}\subset\mathcal{\mathcal{A}}_{N}$,
mapped to the class ~$D'^{m}_{n E}(N)$ by embedding their associated
tensor fields\footnote{As previously noted, the given
~$T^{\mu}_{\nu}\in D'^{m}_{n E}(M)$ we can define in every chart by
$(Ch_{k},\omega)\to\int_{\Omega_{Ch}}
\left[A_{s}(T^{\mu}_{\nu})\right](Ch_{k})~\omega$.} into $N$. This
defines a tensor field $T^{\mu...}_{\nu...}$ living on a piecewise
smooth manifold $N$. Furthermore if
~$\mathcal{\mathcal{\mathcal{\mathcal{A}}}}_{3D}$~ is some
largest piecewise smooth atlas on $N$, the tensor field
~$T^{\mu...}_{\nu...}$~ defines on its domain a map
~$\forall\Omega_{Ch}\subset N$,~ $ Ch_{k}(\Omega_{Ch})\in
\mathcal{\mathcal{A}}_{3D}$ ~and ~$\omega\in C^{P}(\Omega_{Ch})$
\begin{equation}
(\omega,
Ch_{k})\to\int_{\Omega_{Ch}}T^{\mu...}_{\nu...}(Ch_{k})~\omega .
\end{equation}
 This is object from the class ~$D'^{m}_{n E}(N)$. What remains to be proven is that in
any smooth subatlas we map the same ~$D'^{m}_{n E}(M)$ object to
~$D'^{m}_{n}(N)$, otherwise this formulation is meaningless.

\item
The other case is a map ~$D'^{m}_{n (
\mathcal{\mathcal{\mathcal{\mathcal{S}}}} o)}(M)\to D'^{m}_{n ~(
\mathcal{\mathcal{\mathcal{\mathcal{S}}}} o)}(N)$ ($
\mathcal{\mathcal{\mathcal{\mathcal{S}}}}\subset\mathcal{\mathcal{A}}_{N}$),~
defined in a simple way: The objects from ~$D'^{m}_{n S (
\mathcal{\mathcal{\mathcal{\mathcal{S}}}} o)}(M)$ are mapped as
associated smooth tensor fields (in the previous sense). After this
step is taken, one maps the rest of distributional objects from
~$D'^{m}_{n ( \mathcal{\mathcal{\mathcal{\mathcal{S}}}} o)}(M)$
by using the fact that they are weak limits of smooth distributions
~$D'^{m}_{n S ( \mathcal{\mathcal{\mathcal{\mathcal{S}}}}
o)}(M)$. (This is coordinate independent for arbitrary tensor
distributions.) So in the case of objects outside the class
~$D'^{m}_{n S ( \mathcal{\mathcal{\mathcal{\mathcal{S}}}}
o)}(M)$ we embed the smooth distributions first, and take the limit
afterwards (exchanging the order of operations). The basic
conjecture is that if this limit exists on $M$, it will exist on $N$
(in the weak topology), by using the embedded smooth distributions.
\end{itemize}

Now we can say that initial value conditions of, (for example),
second-order partial~ differential ~equations ~are ~given ~by ~two
~distributional ~objects ~from ~~$D'^{m}_{n ( \mathcal{\mathcal{\mathcal{\mathcal{S}}}} o)}(N_{1})$,
~$D'^{m}_{n (\mathcal{\mathcal{\mathcal{\mathcal{S}}}} o)}(N_{2})$~ (~$D'^{m}_{n E}(N_{1})$, ~$D'^{m}_{n
E}(N_{2})$~)~ on two hypersurfaces ~$N_{1}, N_{2}$~ (not
intersecting each other). The solution is a distributional object
from the same class, which fulfills the ~$\approx$ equation and is
mapped (by the maps introduced in this section) to these two
distributional objects.

\paragraph{Useful conjecture related to our approach}

Note, that we can possibly (if the limit commutes) extend this
``initial value'' approach through the ~$D'^{m}_{n S}(M)$ class to
all the weak topology limits of the sequences formed by the objects from this class.
This means extension to the class of objects belonging to
~$D'^{m}_{n}(M)$, such that for any chart from $\mathcal{A}$ they have the full
~$C^{P}(M)$ domain and ~$D'^{m}_{n S}(M)$ is
dense in this class.

\subsubsection{Some additional ideas}

\paragraph{``Null geodesic solution'' conjecture}

Let us conjecture the following:

\theoremstyle{definition}
\newtheorem{conj7}{Conjecture}[section]
\begin{conj7}\label{conj7}
Pick some atlas $\mathcal{\mathcal{\mathcal{\mathcal{\mathcal{S}}}}}$. Pick some ~$g_{\mu\nu}\in D^{0}_{2 S (\mathcal{\mathcal{\mathcal{\mathcal{S}}}}
o)}(M)$, such that it has  $A_{s}(g_{\mu\nu})$, being a
Lorentzian signature metric tensor field. Take some $\Omega_{Ch}$
and two spacelike hypersurfaces $H_{1},~H_{2}$, ~$H_{1}\cap
H_{2}=\{0\}$,
~$H_{1}\cap\Omega_{Ch}\neq\{0\}$,~$H_{2}\cap\Omega_{Ch}\neq\{0\}$. ~Furthermore ~$H_{1}, H_{2}$~ are such that there exist two points ~$q_{1}\in
H_{1}\cap\Omega_{Ch}$, ~$q_{2}\in H_{2}\cap\Omega_{Ch}$ separated by
a null curve geodesics (relatively to ~$A_{s}(g_{\mu\nu})$), and the geodesics lies within
~$\Omega_{Ch}$. Construct such chart ~$Ch_{k}(\Omega_{Ch})\in \mathcal{\mathcal{\mathcal{\mathcal{\mathcal{S}}}}}$,
that both of the hypersurfaces are hypersurfaces
(they are smooth manifolds relatively to $\mathcal{\mathcal{\mathcal{\mathcal{\mathcal{S}}}}}$) given by $u=const.$
condition ($u$ is one of the coordinates) and the given geodesics is representing $u$-coordinate curve.

Take classical free field equation with equivalence:
~$\Box_{g}\Phi\approx 0$~ (with $g^{\mu\nu}$ as previously defined).
Then look for the distributional solution of this equation with the
initial value conditions being ~$\delta\big(Ch_{1k}(\Omega_{Ch}\cap
H_{1}),q_{1}\big)\in D'^{m}_{n (\mathcal{\mathcal{\mathcal{\mathcal{S}}}} o)}(H_{1})$~ on the first
hypersurface and ~$\delta\big(Ch_{2k}(\Omega_{Ch}\cap H_{2}),q_{2}\big)\in
D'^{m}_{n(\mathcal{\mathcal{\mathcal{\mathcal{S}}}} o)}(H_{2}) $~ on the second hypersurface. (Here
~$Ch_{1k}(\Omega_{Ch}\cap H_{1}),~Ch_{2k}(\Omega_{Ch}\cap H_{2})$
are coordinate charts, which are the same on the intersection of the
given hypersurface and ~$\Omega_{Ch}$ as the original coordinates
without $u$.) Then the solution of this initial value problem is a
mapping ~$\Phi$, which is such, that it can be in chart
~$Ch_{k}(\Omega_{Ch})$~ expressed as
\[\omega~\to~\int
du\int\prod_{i}dx^{i}\left(\delta(x^{i}(q_{1}))\omega'(Ch_{k})(u,x^{j})\right)\]
(where
$x^{i}(q_{1})$ is image of $q_{1}$ in chart mapping
~$Ch_{k}(\Omega_{Ch})$ ).
\end{conj7}

We can formulate similar conjectures for timelike and spacelike
geodesics, we just have to:
\begin{itemize}
\item
instead of point separation by null curve, consider the
separation by timelike or spacelike curve,
\item
instead of the ``massless'' equation we have to solve
the ~$(\Box_{g}\pm m^{2})\Phi\approx 0$~ equation ($m$ being
arbitrary nonzero real number). Here $\pm$ depends on the
signature we use and on whether we look for timelike or spacelike
geodesics.
\end{itemize}

The rest of the conditions are unchanged (see \ref{conj7}). Some
insight to our conjectures can be brought by calculating the
massless case for flat Minkowski space, using modified cartesian
coordinates ($u=x-ct,x,y,z$). We get the expected results.\newline

\section{What previous results can be recovered, and how?}

As was already mentioned, our approach is in some sense a generalization of
Colombeau approach from \cite{multi}, which is equivalent to
canonical $\mathbb{R}^{n}$ approach. So for $\Omega_{Ch}$, after
we pick some ~ $Ch'(\Omega_{Ch})\in\mathcal{\mathcal{\mathcal{\mathcal{\mathcal{S}}}}}$, ~(which determines the classes
$A^{n}(M)$ related to this chart), and by considering
only the objects ~$C^{P}_{S (\mathcal{\mathcal{\tilde S}})}(\Omega_{Ch})$ ~($\mathcal{\mathcal{\mathcal{\mathcal{S}}}}\subset\mathcal{\mathcal{\tilde S}}$),~ (hence considering the ~$D'^{m}_{n (\mathcal{\mathcal{\mathcal{\mathcal{S}}}} o)}(M)$ class only), we
obtain from our construction the mathematical language used in
\cite{multi}. But all the basic equivalence relations from Colombeau
approach have been generalized first to the class ~$D'^{m}_{n  (\mathcal{\mathcal{\mathcal{\mathcal{S}}}} o)A}(M)$ and also to appropriate subclasses of the ~$D'^{m}_{n E A}(M)$
class.

\subsection{Generalization of some particular statements}

Now there are certain statements in $\mathbb{R}^{n}$, where one has
to check whether they are just a result of this specific
reduction, or not. A good example is a statement
\begin{equation}
H^{n}~\delta\approx\frac{1}{n+1}~\delta~.
\end{equation}
($H$ is Heaviside distribution.) What we have to do is to interpret
the symbols inside this equation geometrically. This is a
$\mathbb{R}^{1}$ relation. $H$ is understood as a ~$D'_{(\mathcal{\mathcal{\mathcal{\mathcal{S}}}} o)}(M)$
element and defined on the manifold (one dimensional, so the
geometry would be quite trivial) by integral given by a function (on
$M$) obtained by the inverse coordinate mapping substituted to $H$. Now take some fixed chart ~$Ch_{k}(\mathbb{R}^{1})$.
~The derivative is a covariant derivative along the smooth vector
field ~$U$,~ which is constant and unit in the fixed coordinates ~$Ch_{k}(\mathbb{R}^{1})$.
Then $\delta$ can be reinterpreted as ~$\delta(Ch_{k},q)$, where $q$ is
the 0 point in the chart ~$Ch_{k}(\mathbb{R}^{1})$.~ Then the relation can be generalized,
since it is obvious that (see the covariant derivative section)
~$D_{C(U)}H=\delta(q,Ch_{k})$ and so
\begin{eqnarray}
D_{C(U)}H=D_{C(U)}L(H_{f})~~~~~~~~~~~~~~~~~~~~~~~~~~~~~~~~~~~~~~~~~~~~~~~~~~~~~~~~~~~~~~~~~~~~~~~~~~~~~~~~~~~~~~~~~~\nonumber\\
=D_{C(U)}L(H_{f}^{n+1})\approx D_{C(U)}(H^{n+1})=(n+1)~H^{n}~D_{C(U)}H.~~~~
\end{eqnarray}
By $L$ we mean here a regular distribution defined by the function in the
brackets, hence an object
from $D'_{E}(M)$. ~(To be precise and to avoid confusion in the notation, we
used for the Heaviside function the symbol $H_{f}$, while for the Heaviside distribution the usual symbol $H$.) This is nice, but rather trivial
illustration.

This can be generalized to more nontrivial cases.  Take the flat
$\mathbb{R}^{n}$ topological manifold. Fix such chart ~$Ch_{k}(\mathbb{R}^{n})$~ covering the
whole manifold, that we can express Heaviside distribution in this chart through
~$H_{f}(x_{1})\in D'_{(\mathcal{\mathcal{\mathcal{\mathcal{S}}}} o)}(M)$.~ This means the hypersurface where
~$H_{f}$ is ~discontinuous is given in ~$Ch_{k}(\mathbb{R}^{n})$~ as $x_{1}=0$. Now the derivative will be a covariant derivative taken along a smooth vector
field being perpendicular (relatively to the flat space metric\footnote{Note
that since it is flat space it makes sense to speak about a
perpendicular vector field, since we can uniquely transport vectors
to the hypersurface.}) to the hypersurface on which is $H_{f}$ discontinuous. We easily see that
the covariant derivative of $H$ along such vector field gives a distribution (call it
~$\tilde\delta\in D'_{(\mathcal{\mathcal{\mathcal{\mathcal{S}}}} o)}(M)$), which is in the chart $Ch_{k}(\mathbb{R}^{n})$ expressed
as
\begin{equation}
\tilde\delta(\phi)=\delta_{x_{1}}\left(\int\phi(x_{1},x_{2}\dots x_{n})
~dx_{2}\dots dx_{n}\right)~.
\end{equation}
This distribution reminds us in some sense the
``geodesic'' distribution from the previous part. Then the
following holds:
\begin{equation}
H^{n}~\tilde\delta\approx\frac{1}{n+1}~\tilde\delta~.
\end{equation}
This generalized form of our previous statement can be used for computations with Heaviside functional metrics (computation
of connection in fixed coordinates).

\subsection{Relation to practical computations}

These considerations (for example) imply that the result from
canonical Co\-lom\-beau $\mathbb{R}^{n}$ theory derived by
\cite{multi} can be derived in our formalism as well. This is also
true for the geodesic computation in curved space geometry from
\cite{geodesic}. The results derived in special Colombeau algebras
(in geometrically nontrivial cases) are more complicated, since in such cases the strongest, ~$A^{\infty}(\mathbb{R}^{n})$,~ version
of the theory is used. This version is not contained in our chart representations.
(This is because we are using only ~$A^{n}(\mathbb{R}^{n})$ with finite $n$.)
It is clear that all the equivalence relations from our theory must hold in
such stronger formulation (since obviously ~$A^{\infty}(M)\subset
A^{n}(M)~~\forall n\in\mathbb{N}$) and the uniqueness of
distribution solution must hold as well. This means that at this
stage there seems to be no obstacle to reformulate our theory by
using ~$A^{\infty}(Ch_{k},q,\Omega_{Ch})$ classes (and taking
elements from ~$D'_{(\mathcal{\mathcal{\tilde S}} o)}(M)$ at least), if necessary. But
it is unclear whether one can transfer all the calculations using
~$A^{\infty}(Ch_{k},q,\Omega_{Ch})$ classes to our weaker
formulation.

The strong formulation was used also in the Schwarzschild case
\cite{Schwarzschild}, but there is a problem. The fact that the authors of \cite{Schwarzschild}
regularize various functions piece by piece does not have to be necessarily a
problem in Colombeau theory\footnote{Although the authors use (in the
first part, not being necessarily connected with the results) quite
problematic embeddings.}. But, as already mentioned, the problem lies in the use of formula for
~$R^{~~\nu}_{\mu}$, originally derived within smooth tensor field algebra. If we want to derive in the Schwarzschild case Ricci tensor
straight from its definition, we cannot avoid multiplications of
delta function by a non-smooth function. This is in Colombeau
theory deeply non-trivial.

In the cases of Kerr's geometry and conical spacetimes theory this problem appears as well. As a
consequence of this fact, calculations are mollifier dependent, not being (in the strict sense)
results of our theory anymore. On the other hand there is
no reasonable mathematical theory in which these calculations make sense.
This means that a better understanding of these results will be
necessary. By better understanding of these results provided by our theory we mean
their derivation by a net of equivalence relations, by taking
some intermediate quantities to be nonlinear. So the results should follow from
the principle that the equivalence relations are the
fundamental part of all the mathematical formulation of physics.

\section{Conclusions}

The main objectives of this work were to build foundations of a
mathematical language reproducing the old language of
smooth tensor calculus and extending it at the same time. The reasons for these objectives were given at the
beginning of this paper. This work is a first step to such theory,
but it already achieves its  basic goals. That means we consider these results as useful independently of how successful future work on the topic will turn out to be. On the other hand, the territory it opens for further exploration is in my opinion large and significant. It offers a large area of possibilities for future work.

Just to summarize: the result of our work is a theory based purely
on equivalence relations instead of equalities, using a well defined concept
of generalized tensor field and the covariant derivative operator. This operator
is well defined at least on the proper subclass of generalized
tensor fields. We also defined (using some conjectures) the initial value problem for partial differential
equivalence relations. Our theory naturally relates to many results beyond
the classical smooth tensor calculus, already derived.

\end{document}